%% file: arxiv.tex
\newif\ifdraft
 \draftfalse

\documentclass[11pt]{article}

\usepackage[letterpaper,margin=1in]{geometry}
\usepackage[parfill]{parskip}

\usepackage{authblk}

\usepackage[toc,page,header]{appendix}

\usepackage[utf8]{inputenc} 
\usepackage[T1]{fontenc}    
\usepackage[colorlinks,citecolor=blue,urlcolor=blue,linkcolor=blue,linktocpage=true]{hyperref}       
\usepackage{url}            
\usepackage{booktabs}       
\usepackage{amsfonts}       
\usepackage{nicefrac}       
\usepackage{microtype}      
\usepackage{xcolor}         
\label{key}
\usepackage{thmtools}

\usepackage{comment}
\usepackage{amsmath}
\usepackage{epsfig}
\usepackage{graphicx}
\usepackage{wrapfig}
\usepackage{subfigure}
\usepackage{multirow}
\usepackage{hyperref}
\usepackage{graphicx}
\usepackage{caption}
\usepackage{array}
\usepackage{bbm}
\usepackage{color}
\usepackage{enumerate}
\usepackage{enumitem}
\setlist{leftmargin=10mm}
\usepackage{mathtools}
\usepackage{amsmath,amssymb,amsthm,bm}
\usepackage{amsfonts,graphicx}
\usepackage{mathrsfs}

\usepackage{algorithm}
\usepackage[noend]{algpseudocode}

\input{macro}

\input{math_command}

\author[1]{Jiachen T. Wang$^{\mathbf{\star}}$}
\author[1]{Prateek Mittal}
\author[2]{Ruoxi Jia$^{\mathbf{\star}}$}

\affil[1]{Princeton University}

\affil[2]{Virginia Tech\protect\\
\texttt{\small \{tianhaowang,pmittal\}@princeton.edu},
\texttt{\small ruoxijia@vt.edu}
}

\date{}

\title{Efficient Data Shapley for Weighted Nearest Neighbor Algorithms}

\begin{document}



\maketitle

\newcommand\blfootnote[1]{%
  \begingroup
  \renewcommand\thefootnote{}\footnote{#1}%
  \addtocounter{footnote}{-1}%
  \endgroup
}

\begin{abstract}
This work aims to address an open problem in data valuation literature concerning the efficient computation of Data Shapley for weighted $K$ nearest neighbor algorithm (WKNN-Shapley). By considering the accuracy of hard-label KNN with discretized weights as the utility function, we reframe the computation of WKNN-Shapley into a counting problem and introduce a quadratic-time algorithm, presenting a notable improvement from $O(N^K)$, the best result from existing literature. We develop a deterministic approximation algorithm that further improves computational efficiency while maintaining the key fairness properties of the Shapley value. Through extensive experiments, we demonstrate WKNN-Shapley's computational efficiency and its superior performance in discerning data quality compared to its unweighted counterpart.\blfootnote{$^\mathbf{\star}$Correspondence to \textbf{Jiachen T. Wang} and \textbf{Ruoxi Jia}.}
\end{abstract}

\section{Introduction}

Data is the backbone of machine learning (ML) models, but not all data is created equally. In real-world scenarios, data often carries noise and bias, sourced from diverse origins and labeling processes \cite{northcutt2021pervasive}. Against this backdrop, data valuation emerges as a growing research field, aiming to quantify the quality of individual data sources for ML training. Data valuation techniques are critical in explainable ML to diagnose influential training instances and in data marketplaces for fair compensation. The importance of data valuation is highlighted by legislative efforts \cite{dashboardact} and vision statements from leading tech companies \cite{openaifuture}. 
For instance, OpenAI listed ``\emph{how to fairly distribute the benefits the AI systems generate}'' as an important question to be explored.


\textbf{Data Valuation via the Shapley Value.} 
Drawing on cooperative game theory, the technique of using the Shapley value for data valuation was pioneered by \cite{ghorbani2019data,jia2019towards}. 
The Shapley value is a renowned solution concept in game theory for fair profit attribution \cite{shapley1953value}. 
In the context of data valuation, individual data points or sources are regarded as ``players'' in a cooperative game, and \emph{Data Shapley} refers to the suite of data valuation techniques that use the Shapley value as the contribution measure for each data owner. 


\textbf{KNN-Shapley.}
Despite offering a principled approach to data valuation with a solid theoretical foundation, the exact calculation of the Shapley value has the time complexity of $O(2^N)$ in general, where $N$ refers to the number of data points/sources. 
Fortunately, \cite{jia2019efficient}  discovered an efficient $O(N \log N)$ algorithm to compute the exact Data Shapley for unweighted K-Nearest Neighbors (KNN), one of the oldest yet still popular ML algorithms. \emph{KNN-Shapley} refers to the technique of assessing data value for any learning algorithms based on KNN's Data Shapley score. Here, KNN serves as a proxy model for the original, perhaps complicated learning algorithm. KNN-Shapley can be applied to large, high-dimensional datasets by calculating the value scores on the features extracted from neural network embeddings. Due to its superior efficiency and effectiveness in discerning data quality, KNN-Shapley has become one of the most popular data valuation techniques \cite{pandl2021trustworthy}.

\textbf{Open question from \cite{jia2019efficient}: efficient computation of weighted KNN-Shapley.} 
While \cite{jia2019efficient} showed unweighted KNN-Shapley can be computed efficiently, they did not develop a practical algorithm for the more general weighted KNN-Shapley (WKNN-Shapley). Despite presenting a polynomial-time algorithm, its computational complexity of $O(N^K)$ becomes impractical even for modest $K$, such as 5.
Closing this efficiency gap is important, especially given the inherent advantages and wider application of weighted KNN. Compared with the unweighted counterpart, weighted KNN considers the distance between data points, assigning varying levels of importance to neighbors based on their proximity to the query. Consequently, the Data Shapley of weighted KNN can potentially better discriminate between low and high-quality data points. Moreover, weighted KNN is applied more widely in practice. For example, recent research discovered weighted KNN's capability to improve language model's performance \cite{khandelwal2019generalization}. 

Our contributions are summarized as follows:

\textbf{Making KNN Configurations ``Shapley-friendly'' (Section \ref{sec:challenge}).} 
Our preliminary investigations suggest that improving the computational efficiency of WKNN-Shapley for soft-label KNN classifiers with continuous weight values (the setting considered in \cite{jia2019efficient}), poses considerable challenges. Consequently, we make necessary modifications to the specific KNN classifiers' configuration and shift our focus to \emph{hard-label} KNN classifiers with \emph{discrete weight values}. The justification for these changes and their practical relevance is detailed in Section \ref{sec:challenge}. In particular, discretizing weights does not change the data value score significantly even with few bits. We emphasize that making proper tweaks to the problem setup is important for developing an efficient Shapley computation algorithm, a strategy frequently adopted in literature \cite{dall2019sometimes}.

\textbf{A quadratic-time algorithm for computing exact WKNN-Shapley \add{(Section \ref{sec:exact-shapley})}.} 
\add{Given the adjusted ``Shapley-friendly'' configurations of the weighted KNN,} we can reframe the Shapley value computation as a counting problem. We develop an algorithm with a quadratic runtime for solving the counting problem and computing the exact WKNN-Shapley, \add{which greatly improves the baseline $O(N^K)$ algorithm.}

\textbf{A \add{subquadratic-time} 
deterministic approximation algorithm that preserves fairness properties \add{(Section \ref{sec:deterministic-approx})}.} 
To further improve the computational efficiency, we propose a deterministic approximation algorithm \add{by making minor changes to the exact WKNN-Shapley implementation.} 
In particular, our approximation algorithm retains the crucial fairness properties of the original Shapley value. 
\textbf{Empirical Evaluations \add{(Section \ref{sec:eval})}.} 
We experiment on benchmark datasets and assess the efficiency and efficacy of our exact and approximation algorithms for WKNN-Shapley. 
Here are the key takeaways: 
\textbf{(1)} Our exact and approximation algorithm for WKNN-Shapley significantly improves computational efficiency compared to the baseline exact algorithm and Monte Carlo approximation, respectively. 
\textbf{(2)} WKNN-Shapley outperforms the unweighted KNN-Shapley in discerning data quality for critical downstream tasks such as detecting mislabeled/noisy data and data selection. \add{Remarkably, the approximated WKNN-Shapley matches the performance of the exact WKNN-Shapley on many benchmark datasets, attributable to its deterministic nature and the preservation of fairness properties.}


\add{
Overall, with proper changes to KNN configurations, we show that WKNN-Shapley can be efficiently calculated and approximated. This facilitates its wider adoption, offering a more effective data valuation method compared to unweighted KNN-Shapley.
}

\section{Preliminaries}

In this section, we formalize the setup of data valuation for ML, and revisit relevant techniques. 


\textbf{Setup \& Goal.} Given a labeled dataset $D := \{ z_i \}_{i=1}^N$ where each data point $z_i := (x_i, y_i)$, data valuation aims to assign a score to each training data point $z_i$, reflecting its importance for the trained ML model's performance. Formally, we seek a score vector $(\phi_{z_i})_{i=1}^N$ where each $\phi_{z_i} \in \R$ represents the ``value'' of $z_i$.



\subsection{Data Shapley}

The Shapley value (SV) \cite{shapley1953value}, originating from game theory, stands out as a distinguished method for equitably distributing total profit among all participating players. Before diving into its definition, we first discuss a fundamental concept: the \emph{utility function}.

\textbf{Utility Function.}
A \emph{utility function} maps an input dataset to a score indicating the utility of the dataset for model training. 
Often, this function is chosen as the validation accuracy of a model trained on the given dataset. That is, given a training set $S$, the utility function $\U(S) := \metric(\A(S))$, where $\A$ represents a learning algorithm that trains a model on dataset $S$, and $\metric(\cdot)$ is a function assessing the model's performance, e.g., its accuracy on a validation set. 

\begin{definition}[Shapley value \cite{shapley1953value}]
\label{def:shapley-value}
Let $\U(\cdot)$ denote a utility function and $D$ represent a training set of $N$ data points. The Shapley value, $\phi_z\left( \U \right)$, assigned to a data point $z \in D$ is defined as 
$
\phi_z\left( \U \right) := \frac{1}{N} \sum_{k=1}^{N} {N-1 \choose k-1}^{-1} \sum_{S \subseteq D_{-z}, |S|=k-1} \left[ \U(S \cup \{z\}) - \U(S) \right]
$ where $D_{-z} = D \setminus \{z\}$. 
\end{definition}

In simple terms, the Shapley value is a weighted average of the \emph{marginal contribution} $\U(S \cup \{z\}) - \U(S)$, i.e., the utility change when the point $z$ is added to different $S$s. For simplicity, we often write $\phi_z$ when the utility function is clear from the context. 
The Shapley value uniquely satisfies several important axioms, including key fairness requirements like the \emph{null player} and \emph{symmetry} axioms, lending justification to its popularity. See Appendix \ref{appendix:related-works} for detailed axiom definitions. 


\subsection{KNN-Shapley}

A well-known challenge of using the Shapley value is that its exact calculation is computationally infeasible in general, as it usually requires evaluating $\U(S)$ for all possible subsets $S \subseteq D$. A surprising result in \cite{jia2019efficient} showed that when the learning algorithm $\A$ is \emph{unweighted KNN}, there exists a highly efficient algorithm for computing its exact Data Shapley score. 

\textbf{K Nearest Neighbor Classifier.} 
Given a validation data point $\zval = (\xval, \yval)$ and a distance metric $d(\cdot, \cdot)$, we sort the training set $D = \{z_i = (x_i, y_i)\}_{i=1}^N$ according to their distance to the validation point $d(x_i, \xval)$ in non-descending order. Throughout the paper, we assume that $d(x_i, \xval) \le d(x_j, \xval)$ for any $i \le j$ unless otherwise specified. A KNN classifier makes a prediction for the query $\xval$ based on the (weighted) majority voting among $\xval$'s $K$ nearest neighbors in the training set. 
\textbf{Weight of data point:} in KNN, each data point $z_i$ is associated with a weight $w_i$. The weight is usually determined based on the distance between $x_i$ and the query $\xval$. For example, a popular weight function is RBF kernel $w_i := \exp(-d(x_i, \xval))$. If $w_i$ is the same for all $z_i$s, it becomes \emph{unweighted} KNN.


\cite{jia2019efficient} considers the utility function for the weighted, \emph{soft-label} KNN: 
\begin{align}
\U(S; \zval) &:= \frac{
\sum_{j=1}^{\min(K, |S|)} w_{\alpha_{\xval}^{(S, j)}} \ind \left[y_{\alpha_{\xval}^{(S, j)}} = \yval \right]
}{ \sum_{j=1}^{\min(K, |S|)} w_{\alpha_{\xval}^{(S, j)}} }
\label{eq:util-weighted-softlabel}
\end{align}
where $\alpha_{\xval}^{(S, j)}$ denotes the index (among $D$) of $j$th closest data point in $S$ to $\xval$. 
\add{``Soft-label'' refers to the classifiers that output the confidence scores.} The main result in \cite{jia2019efficient} shows that for \emph{unweighted} KNN, we can compute the \emph{exact} Shapley value $\phi_{z_i} \left(\U(\cdot; \zval) \right)$ for \emph{all} $z_i \in D$ within a total runtime of $O(N \log N)$ (see Appendix \ref{appendix:knn-background} for details). 

\begin{remark}
\add{In practice, the model performance is assessed based on a validation set $\Dval$. After computing $\phi_{z_i} \left(\U(\cdot; \zval) \right)$ for each $\zval \in \Dval$, one can compute the Shapley value corresponding to the utility function on the full validation set $\U(S; \Dval) := \sum_{\zval \in \Dval} \U(S; \zval)$ by simply taking the sum $\phi_{z_i} \left(\U(\cdot; \Dval) \right) = \sum_{\zval \in \Dval} \phi_{z_i} \left(\U(\cdot; \zval) \right)$ due to the \emph{linearity} property of the Shapley value.}
\end{remark}

\begin{remark}
\add{In alignment with the existing literature \cite{jia2019efficient, wang2023threshold}, our discussion of the time complexity of KNN-Shapley refers to the total runtime needed to calculate \emph{all} data value scores $\left(\phi_{z_1}(\U(\cdot; \zval)), \ldots, \phi_{z_N}(\U(\cdot; \zval))\right)$, given that standard applications of data valuation, such as profit allocation and bad data detection, all require computing the data value scores for \emph{all} data points in the training set.
Furthermore, the stated runtime is with respect to $\U(\cdot; \zval)$, and the overall runtime with respect to $\U(\cdot; \Dval)$ will be multiplied by the size of validation set $\Dval$. 
Runtime is typically presented this way because 
SV computations with respect to different $\U(\cdot; \zval)$ are independent, readily benefit from parallel computing.} 
\label{remark:runtime}
\end{remark}

Since its introduction, KNN-Shapley has rapidly gained popularity in data valuation for its efficiency and effectiveness, and is being advocated as the \emph{`most practical technique for effectively evaluating large-scale data'} in recent studies \cite{pandl2021trustworthy, karlavs2022data}.



\subsection{Baseline Algorithm for Computing and Approximating WKNN-Shapley}
\label{sec:baseline-algorithm}

\cite{jia2019efficient} developed an efficient $O(N \log N)$ algorithm to calculate the exact unweighted KNN-Shapley. However, when it comes to the more general weighted KNN-Shapley, only an $O(N^K)$ algorithm is given. 
While still in polynomial time (if $K$ is considered a constant), the runtime is impractically large even for small 
$K$ (e.g., 5). 
Here, we review the high-level idea of the baseline algorithms from \cite{jia2019efficient}. 



\textbf{An $O(N^K)$ algorithm for exact WKNN-Shapley computation.} 
From Definition \ref{def:shapley-value}, the Shapley value for $z_i$ is a weighted average of the \emph{marginal contribution (MC)} $\U(S \cup \{z_i\}) - \U(S)$; hence, we only need to study those $S$ whose utility might change due to the inclusion of $z_i$. For KNN, those are the subsets $S$ where $z_i$ is within the $K$ nearest neighbors of $\xval$ after being added into $S$. Note that for KNN, the utility of any $S$ only depends on the $K$ nearest neighbors of $\xval$ in $S$. Given that there are only $\sum_{j=0}^K {N \choose j}$ unique subsets of size $\le K$, we can simply query the MC value $\U(S \cup \{z_i\}) - \U(S)$ for all $S$ of size $\le K$. For any larger $S$, the MC must be the same as its subset of $K$ nearest neighbors. We can then compute the Shapley value as a weighted average of these MC values by counting the number of subsets that share the same MC values through simple combinatorial analysis. Such an algorithm results in the runtime of $\sum_{j=0}^K {N \choose j} = O(N^K)$. The algorithm details can be found in Appendix \ref{appendix:knn-background}. 

\textbf{Monte Carlo Approximation.} 
Given the large runtime of this exact algorithm, \cite{jia2019efficient} further proposes an approximation algorithm based on Monte Carlo techniques. However, Monte Carlo-based approximation is randomized and may not preserve the fairness property of the exact Shapley value. Additionally, the sample complexity of the Monte Carlo estimator is derived from concentration inequalities, which, while suitable for asymptotic analysis, may provide loose bounds in practical applications.


\section{Making KNN Configurations Shapley-friendly}
\label{sec:challenge}

We point out the major challenges associated with directly improving the computational efficiency for the soft-label KNN configuration considered in \cite{jia2019efficient}, and propose proper changes that enable more efficient algorithms for computing WKNN-Shapley. 

\textbf{Challenge \#1: weights normalization term.} The key principle behind the $O(N \log N)$ algorithm for unweighted KNN-Shapley from \cite{jia2019efficient} is that, \add{the MC can only take few distinct values. 
For example, for any $|S| \ge K$, we have 
$
\U(S \cup \{z_i\}) - \U(S) = \frac{1}{K} \left( \ind[y_{i} = \yval] - \ind[y_{\alpha_{\xval}^{(S, K)}} = \yval] \right)
$.
To avoid the task of evaluating $\U(S)$ for all $S \subseteq D$, one can just count the subsets $S \subseteq D \setminus \{z_i\}$ such that $z_i$ is among the $K$ nearest neighbors of $\xval$ in $S \cup \{z_i\}$, as well as the subsets share the same $K$th nearest neighbor to $\zval$. 
However, for weighted soft-label KNN with the utility function in (\ref{eq:util-weighted-softlabel}), there is little chance that any of two 
$\U(S_1 \cup \{z_i\})-\U(S_1)$ and 
$\U(S_2 \cup \{z_i\})-\U(S_2)$
can have the same value due to the weights normalization term $\left(\sum_{j=1}^{K} w_{\alpha_{\xval}^{(S, j)}}\right)^{-1}$ (note that this term is $1/K$, a constant, for unweighted setting).} 
\textbf{Solution \#1: hard-label KNN.} 
In this work, we instead consider the utility function for weighted \emph{hard-label} KNN. ``Hard-label'' refers to the classifiers that output the predicted class instead of the confidence scores (see (\ref{eq:util}) in Section \ref{sec:shapley-for-binary}). In practice, user-facing applications usually only output a class prediction instead of the entire confidence vector. More importantly, hard-label KNN's prediction only depends on the weight comparison between different classes, and hence its utility function does not have a normalization term. 

\textbf{Challenge \#2: continuous weights.} If the weights are on the continuous space, there will be infinitely many possibilities of weighted voting scores of the $K$ nearest neighbors. 
This makes it difficult to analyze which pairs of $S_1, S_2$ share the same MC value. 
\textbf{Solution \#2: discretize weights.} 
Therefore, we consider a more tractable setting where the weights lie in a discrete space. Such a change is reasonable since the weights are stored in terms of finite bits (and hence in the discrete space) in practice. 
Moreover, rounding is a deterministic operation and does not reverse the original order of weights.
In Appendix \ref{appendix:error-disc}, we show that the Shapley value computed based on the discrete weights has the same ranking order compared with the Shapley value computed on the continuous weights (it might create ties but will not reverse the original order). In Appendix \ref{sec:eval-discretization}, we empirically verify that weight discretization does not cause a large deviation in the Shapley value. 

\add{We emphasize that, proper adjustments to the underlying utility function are important for efficient Shapley computation, and such a strategy is frequently applied in game theory literature \cite{dall2019sometimes}.}




\section{Data Shapley for Weighted KNN}
\label{sec:shapley-for-binary}

\add{In this section, we develop efficient solutions for computing and approximating Data Shapley scores for weighted, hard-label KNN binary classifiers, where the weight values used in KNN are discretized. 
Without loss of generality, in this paper we assume every weight $w_i \in [0, 1]$.\footnote{
If it does not hold one can simply normalize the weights to $[0, 1]$ and the KNN classifier remains the same.} 
We use $\wspace$ to denote the discretized space of $[0, 1]$, where we create $2^b$ equally spaced points within the interval when we use $b$ bits for discretization. We denote $\nwspace := |\wspace| = 2^b$ the size of the weight space. 
}

\textbf{Utility Function for Weighted Hard-Label KNN Classifiers.} 
The utility function of weighted hard-label KNN, i.e., the correctness of weighted KNN on the queried example $\zval$, can be written as 
\begin{equation}
\begin{aligned}
\U(S; \zval) = \ind \bigg[ \yval \in \argmax_{ c \in \C }  \sum_{j=1}^{\min(K, |S|)} w_{\alpha_{\xval}^{(S, j)}} \times \ind[y_{\alpha_{\xval}^{(S, j)}} = c] \bigg]
\end{aligned}
\label{eq:util}
\end{equation}
where $\C = \{1, \ldots, C\}$ is the space of classes, and $C$ is the number of classes\footnote{If multiple classes have the same top counts, we take the utility as 1 as long as $\yval$ is among the majority classes.}. 
We omit the input of $\zval$ and simply write $\U(S)$ when the validation point is clear from the context. 
For KNN binary classifier,
we can rewrite the utility function in a more compact form:
\begin{equation}
\begin{aligned}
\U(S) = \ind \left[ \sum_{j=1}^{\min(K, |S|)} \wtil_{\alpha_{\xval}^{(S, j)}} \ge 0 \right]~\text{where}~\wtil_{j} := 
\begin{cases}
    w_{j} & y_{j} = \yval \\
    -w_{j} & y_{j} \ne \yval \\
\end{cases}
\label{eq:util-binary}
\end{aligned}
\end{equation}
\add{For ease of presentation, we present the algorithms for KNN binary classifier here, and defer the extension to multi-class classifier to Appendix \ref{appendix:multiclass}.}

\subsection{Exact WKNN-Shapley Calculation}
\label{sec:exact-shapley}

\subsubsection{Computing SV is a Counting Problem}
\label{sec:SV-counting-problem}

Given that the Shapley value is a weighted average of the marginal contribution $\U(S \cup \{z_i\}) - \U(S)$, we first study the expression of $\U(S \cup \{z_i\}) - \U(S)$ for a fixed subset $S \subseteq D \setminus \{z_i\}$ with the utility function in (\ref{eq:util-binary}). 

\begin{theorem} 
\label{thm:marginal-contri}
For any data point $z_i \in D$ and any subset $S \subseteq D \setminus \{z_i\}$, the marginal contribution has the expression as follows:
\begin{equation}
\begin{aligned}
\U(S \cup \{z_i\}) - \U(S) =
\begin{cases} 
1 & \text{if~}y_i = \yval, \condKNN, \condNegToPos \\
-1 & \text{if~}y_i \ne \yval, \condKNN, \condPosToNeg \\
0 & \text{Otherwise}
\end{cases}
\end{aligned}
\end{equation}
where 
\begin{equation}
\begin{aligned}
\condKNN &:= z_i~\text{is within}~K~\text{nearest neighbors of}~\xval~\text{among}~S \cup \{z_i\} \\
\condNegToPos &:= 
\begin{cases}
    \sum_{z_j \in S} \wtil_j \in [-\wtil_i, 0) & \text{if}~|S| \le K-1 \\
    \sum_{j=1}^{K-1} \wtil_{\alpha_{\xval}^{(S, j)}} \in \left[ -w_i,  -\wtil_{\alpha_{\xval}^{(S, K)}}\right) & \text{if}~|S| \ge K
\end{cases} \\
\condPosToNeg &:= 
\begin{cases}
    \sum_{z_j \in S} \wtil_j \in [0, -\wtil_i) & \text{if}~|S| \le K-1 \\
    \sum_{j=1}^{K-1} \wtil_{\alpha_{\xval}^{(S, j)}} \in \left[-\wtil_{\alpha_{\xval}^{(S, K)}}, -w_i\right) & \text{if}~|S| \ge K
\end{cases}
\end{aligned}
\nonumber
\end{equation}
\end{theorem}
In words, the condition $\condKNN$ means that $z_i$ should be among the $K$ nearest neighbors to the query sample when it is added to $S$. The conditions $\condNegToPos$ and $\condPosToNeg$ cover situations where adding $z_i$ to the set $S$ changes the prediction of the weighted KNN classifiers. 
\add{In greater detail, $\condNegToPos$ captures the condition for which incorporating $z_i$ shifts the ``\emph{effective sum of signed weights $\wtil$}'' from a negative to a non-negative value, thereby incrementing the utility from $0$ to $1$. $\condPosToNeg$ can be interpreted similarly.} 
From Theorem \ref{thm:marginal-contri} and the formula of the Shapley value (Definition \ref{def:shapley-value}), we can reframe the problem of computing hard-label WKNN-Shapley as a counting problem. Specifically, this involves counting the quantity defined as follows:

\begin{definition}
\label{def:Gil}
Let $\Gil$ denote the count of subsets $S \subseteq D \setminus {z_i}$ of size $\ell$ that satisfy \textbf{(1)} $\condKNN$, and \textbf{(2)} $\condNegToPos$ if $y_i = \yval$, or $\condPosToNeg$ if $y_i \ne \yval$. 
\end{definition}

\begin{theorem}
\label{thm:shapley-counting}
For a weighted, hard-label KNN binary classifier using the utility function given by (\ref{eq:util-binary}), the Shapley value of a data point $z_i$ can be expressed as:
\begin{align}
    \phi_{z_i} = 
    \frac{ 2 \ind[y_i = \yval] - 1 }{N} \sum_{\ell=0}^{N-1} {N-1 \choose \ell}^{-1} \Gil 
    \label{eq:shapley-in-gil}
\end{align}
\end{theorem}
\add{Figure \ref{fig:counting-problem} illustrates the counting problem we try to solve here.} 

\begin{figure}[t]
    \centering
    \centering
    \includegraphics[width=0.5\columnwidth]{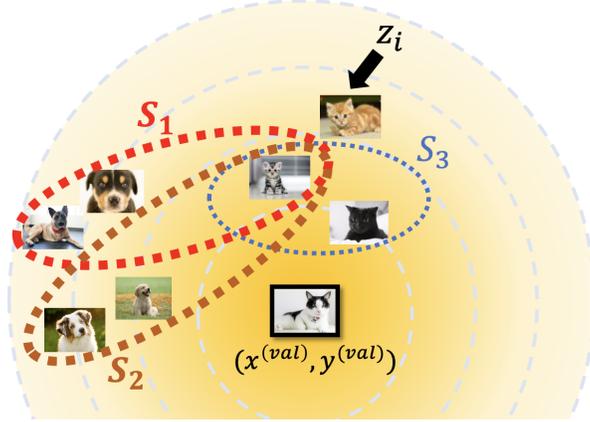}
    \caption{
    Illustration of the subsets targeted in the counting problem. 
    When $K=3$, both $S_1$ and $S_2$ have a utility of 0 as both of them contain 2 dogs and 1 cat. 
    Adding $z_i$ to $S_1$ and $S_2$ alters the $3$ nearest neighbors to the query image $\xval$, which now contains 1 dog and 2 cats, raising the utility to 1. 
    In contrast, $S_3$'s utility remains unchanged with the addition of $z_i$ since it solely contains cat images. To compute WKNN-Shapley of $z_i$, we count the subsets $S$ where adding $z_i$ changes its utility, as seen with $S_1$ and $S_2$.
    }
    \label{fig:counting-problem}
\end{figure}

\subsubsection{Dynamic Programming Solution for Computing $\Gil$}
\label{sec:dynamic-programming}

The multiple, intricate conditions wrapped within $\Gil$'s definition can pose a formidable challenge for direct and efficient counting. 
The main rationale behind our solution is to break down the complex counting problems into smaller, more manageable subproblems, thereby making them amenable to algorithmic solutions like dynamic programming. 
Before delving into the specifics of the algorithm, we introduce an intermediary quantity, $\Fi$, that becomes the building block of our dynamic programming formulation.


\begin{definition}
\label{def:Fi-binary}
Let $\Fi\left[m, \ell, s\right]$ denote the count of subsets $S \subseteq D \setminus \{z_i\}$ of size $\ell$ that satisfy \textbf{(1)} $\condKNN$, as well as the following conditions: 
\textbf{(2)} Within $S$, the data point $x_m$ is the $\min(\ell, K)$-th closest to the query example $\xval$, 
\textbf{(3)} $\sum_{j=1}^{\min(\ell, K-1)} \wtil_{\alpha_{\xval}^{(S, j)}} = s$. 
\end{definition}

We can relate this auxiliary quantity to our desired $\Gil$ as follows:

\begin{theorem}[Relation between $\Gil$ and $\Fi$] 
\label{thm:Gil-and-Fi}
For $y_i = \yval$, we can compute $\Gil$ from $\Fi$ as follows:
\begin{equation}
\begin{aligned}
\Gil = 
\begin{cases} 
\sum_{m \in [N] \setminus i} \sum_{s \in [-\wtil_i, 0)} \Fi\left[m, \ell, s\right] & \text{for } \ell \le K-1, \\
\sum_{m \in [N] \setminus i} \sum_{s \in [-\wtil_i, -\wtil_m)} \Fi\left[m, \ell, s\right] & \text{for } \ell \ge K.
\end{cases}
\end{aligned}
\end{equation}
For $y_i \ne \yval$, we have:
\begin{equation}
\begin{aligned}
\Gil = 
\begin{cases} 
\sum_{m \in [N] \setminus i} \sum_{s \in [0, -\wtil_i)} \Fi\left[m, \ell, s\right] & \text{for } \ell \le K-1, \\
\sum_{m \in [N] \setminus i} \sum_{s \in [-\wtil_m, -\wtil_i)} \Fi\left[m, \ell, s\right] & \text{for } \ell \ge K.
\end{cases}
\end{aligned}
\end{equation}
\end{theorem}

The auxiliary quantity $\Fi$ thus serves as a pivot, allowing us to explore the search space of possible subsets $S$ more systematically. We next exploit the computational advantage of $\Fi$. Specifically, $\Fi$ can be conveniently computed with (recursive) formulas, which further enables us to compute $\Gil$ with reduced computational demand. 

\begin{theorem}[simplified version]
\label{thm:Fi-recursive}
For $\ell \le K-1$, $\Fi[m, \ell, s]$ can be computed from $\Fi[t, \ell, \cdot]$ with $t \le m-1$. 
For $\ell \ge K$, $\Fi[m, \ell, s]$ can be computed from $\Fi[t, K-1, \cdot]$ with $t \le m-1$.
\end{theorem}

Leveraging the results from Theorem \ref{thm:Fi-recursive}, a direct method for calculating $\Gil$ for all $\ell \ge 1$ is as follows: we first use a recursive formula to compute $\Fi[\cdot, \ell, \cdot]$ for $\ell \le K-1$, and then use an explicit formula to compute $\Fi[\cdot, \ell, \cdot]$ for $\ell \ge K$. 
With $\Fi$ being computed, we apply Theorem \ref{thm:Gil-and-Fi} to compute $\Gil$. This direct approach renders an $O(N^3)$ runtime, as it necessitates computing $\Fi[m, \ell, \cdot]$ for each of $i, m,$ and $\ell$ in the range of $1, \ldots, N$.

\textbf{Further Improvement of Efficiency Through Short-cut Formula.} 
While the above direct approach offers a clear path, there exists a more optimized algorithm to expedite computational efficiency by circumventing the explicit calculations of $\Fi[\cdot, \ell, \cdot]$ for $\ell \ge K$. Specifically, we discover a short-cut formula that allows us to directly calculate the summation \add{$\sum_{\ell=K}^{N-1} \frac{\Gil}{{N-1 \choose \ell}}$ from Theorem \ref{thm:shapley-counting}} once we obtain $\Fi[\cdot, K-1, \cdot]$. 

\begin{theorem}
\label{thm:wknn-shapley-expression}
For a weighted, hard-label KNN binary classifier using the utility function given by (\ref{eq:util-binary}), the Shapley value of data point $z_i$ can be expressed as:
\begin{equation}
\begin{aligned}
    \phi_{z_i} 
    = 
    \sign(w_i) \left[
    \frac{1}{N} \sum_{\ell=0}^{K-1} \frac{\Gil}{{N-1 \choose \ell}} + \sum_{m = \max(i+1, K+1)}^N \frac{\Rim}{m {m-1 \choose K}} \right]
\end{aligned}
\nonumber
\end{equation}
where 
\begin{equation}
\begin{aligned}
\Rim := 
\begin{cases}
    \sum_{t=1}^{m-1} \sum_{s \in [-\wtil_i, -\wtil_m)} \Fi[t, K-1, s] & \text{for } y_i = \yval \\
    \sum_{t=1}^{m-1} \sum_{s \in [-\wtil_m, -\wtil_i)} \Fi[t, K-1, s] & \text{for } y_i \ne  \yval 
\end{cases}
\end{aligned}
\nonumber
\end{equation}
\end{theorem}

Crucially, the $\Rim$ quantity in the above expression can be efficiently calculated using a clever caching technique. Based on the above findings, we can eliminate a factor of $N$ in the final time complexity, thereby obtaining a \emph{quadratic-time} algorithm to compute the exact WKNN-Shapley. The comprehensive pseudocode for the full algorithm can be found in Appendix \ref{appendix:pseudocode}.

\begin{theorem}
\label{thm:runtime-exact}
Algorithm \ref{alg:efficient-wknn} (in Appendix \ref{appendix:pseudocode}) computes the exact WKNN-Shapley $\phi_{z_i}$ for all $i = 1, \ldots, N$ and achieves a total runtime of $O(\nwspace K^2 N^2)$.
\end{theorem}

\begin{remark}[Runtime Dependency with $K$ and $\nwspace$]
\add{While the time complexity in Theorem \ref{thm:runtime-exact} also depends on $K$ and $\nwspace$, these variables can be effectively treated as constants in our context. In our ablation study, we found that the error caused by weights discretization reduces quickly as the number of bits $b$ for discretization grows. Hence, across all experiments in Section \ref{sec:eval}, we set the number of bits for discretization as $b = 3$ and therefore $\nwspace = 2^b = 8$. Additionally, the selection of $K$ in KNN-Shapley literature commonly stabilizes around values of $5$ or $10$, irrespective of the dataset size \cite{jia2019efficient, wang2023threshold}. This stability arises because, in the context of KNN classifiers, increasing $K$ can easily result in underfitting even when $N$ is large. Throughout our experiments, we fix $K$ at $5$. A detailed ablation study examining different choices of $K$ and $\nwspace$ is available in Appendix \ref{appendix:eval}.}
\end{remark}

\begin{remark}[Novelty in the derivation of WKNN-Shapley]
Our derivation of WKNN-Shapley starts similarly to \cite{jia2019efficient} by examining the marginal contribution, but \textbf{the methodologies significantly differ afterward.} 
Unweighted KNN-Shapley benefits from the simplicity of its utility function, allowing for a relatively straightforward derivation. Specifically, \cite{jia2019efficient} plugs unweighted KNN's utility function into the formula of the difference of the Shapley values between two data points, and then simplifies the expression using known equalities in combinatorial analysis. 
In contrast, the same approach does not apply to WKNN-Shapley due to the complexity introduced by weights. Therefore, we develop a novel dynamic programming solution, which marks a substantial departure from the techniques used in \cite{jia2019efficient}. We would like to clarify that there is \textbf{no correlation between the ``recursion'' in \cite{jia2019efficient} and our paper.} 
In \cite{jia2019efficient}, each $\phi_{z_{j}}$ is recursively computed from $\phi_{z_{j+1}}$. On the contrary, the computation of different data points' WKNN-Shapley scores $\phi_{z_{j}}$ is \emph{independent} of each other. In our method, recursion is a fundamental component of dynamic programming to solve complex counting problems. 
\end{remark}

\subsection{Deterministic Approximation for Weighted KNN-Shapley}
\label{sec:deterministic-approx}

While the algorithm introduced in Section \ref{sec:exact-shapley} for calculating the exact WKNN-Shapley achieves $O(N^2)$ runtime, a huge improvement from the original $O(N^K)$ algorithm from \cite{jia2019efficient}, there remains room for further improving the efficiency if we only require an approximation of the Shapley value. Contrary to the prevalent use of Monte Carlo techniques in existing literature, in this section, we develop a \emph{deterministic} approximation algorithm for WKNN-Shapley. 

\textbf{Intuition.} 
From Theorem \ref{thm:Fi-recursive}, we know that in order to compute $\Fi[m, \ell, \cdot]$ with $\ell \le K-1$, we only need to know $\Fi[t, \ell-1, \cdot]$ with $t \le m-1$. Moreover, observe that the building blocks for $\Gil$ (or $\Rim$), $\sum_{s \in [-\wtil_i, 0)} \Fi[t, \ell, s]$ (or $\sum_{s \in [-\wtil_i, -\wtil_m)} \Fi[t, K-1, s]$), can be quite small as it only takes the summation over a small range of the weight space. 
Hence, we can use $\Fihat[m, \cdot, \cdot] = 0$ as an approximation for $\Fi[m, \cdot, \cdot]$ for all $m \ge \mstar+1$ with some prespecified threshold $\mstar$. Similarly, we can use $\Rimhat = 0$ as an approximation for $\Rim$ for all $m \ge \mstar+1$. 
The resultant approximation for the Shapley value $\phi_{z_i}$ is stated as follows:

\begin{definition}
\label{def:approxmstar}
We define the approximation $\widehat \phi_{z_i}^{(\mstar)}$ as 
\begin{equation}
\begin{aligned}
\widehat \phi_{z_i}^{(\mstar)} := 
\sign(w_i) \left[
\frac{1}{N} \sum_{\ell=0}^{K-1} \frac{\Giltilmstar}{{N-1 \choose \ell}} + 
\sum_{m = \max(i+1, K+1)}^{\mstar} \frac{\Rim}{m {m-1 \choose K}} \right]
\end{aligned}
\nonumber
\end{equation}
where 
\begin{equation}
\begin{aligned}
\Giltilmstar := 
\begin{cases}
\sum_{m=1}^{\mstar} \sum_{s \in [-\wtil_i, 0)} \Fi\left[m, \ell, s\right] & \text{for } y_i = \yval \\
\sum_{m=1}^{\mstar} \sum_{s \in [0, -\wtil_i)} \Fi\left[m, \ell, s\right] & \text{for } y_i \ne \yval
\end{cases}
\end{aligned}
\nonumber
\end{equation}
\end{definition}

To calculate $\widehat \phi_{z_i}^{(\mstar)}$, we only need to compute $\Fi[m, \cdot, \cdot]$ and $\Rim$ for $m$ from $1$ to $\mstar$ instead of $N$, thereby reducing the runtime of Algorithm \ref{alg:efficient-wknn} to $O(N \mstar)$ with minimal modification to the exact algorithm's implementation. 

\begin{theorem}
\label{thm:runtime-approx}
\add{Algorithm \ref{alg:efficient-wknn-approx} (in Appendix \ref{appendix:pseudocode-approx})} computes the approximated WKNN-Shapley $\widehat \phi_{z_i}^{(\mstar)}$ for all $z_i \in D$ and achieves a total runtime of $O(\nwspace K^2 N \mstar)$.
\end{theorem}

In particular, when $\mstar = \sqrt{N}$, we can achieve the runtime of $O(N^{1.5})$. 
\add{The selection of $\mstar$ is discussed in Remark \ref{remark:mstar-selection}.} 
In the following, we derive the error bound and point out two nice properties of this approximation. 

\begin{theorem}
\label{thm:error-bound}
For any $z_i \in D$, the approximated Shapley value $\phihatmstar$ \textbf{(1)} shares the same sign as $\phi_{z_i}$, \textbf{(2)} ensures $\left| \phihatmstar \right| \le \left| \phi_{z_i} \right|$, and \textbf{(3)} has the approximation error bounded by $\left| \phihatmstar - \phi_{z_i} \right| \le \eps(\mstar)$ where 
\begin{equation}
\begin{aligned}
\eps(\mstar) := \sum_{m = \mstar+1}^N \left( \frac{1}{m-K} - \frac{1}{m} \right) + \sum_{\ell=1}^{K-1} \frac{ {N \choose \ell} - {\mstar \choose \ell} }{ N {N-1 \choose \ell} } = O\left( K / \mstar \right)
\end{aligned} 
\nonumber
\end{equation}
\end{theorem}

Leveraging the error bound $\eps(\mstar)$ alongside the additional nice properties of $\phihatmstar$ stated in Theorem \ref{thm:error-bound}, we can obtain a deterministic interval within which $\phi_{z_i}$ always resides. 
Specifically, when $y_i = \yval$, we have $\phi_{z_i} \in \left[\phihatmstar, \phihatmstar + \eps(\mstar) \right]$, and when $y_i \ne \yval$, we have $\phi_{z_i} \in \left[\phihatmstar - \eps(\mstar), \phihatmstar \right]$. 
Unlike the commonly used Monte Carlo method for approximating the Shapley value, which only offers a high-probability interval and allows for a failure possibility that the exact value might fall outside of it, our deterministic approximation ensures that the exact value is always within the corresponding interval.



\textbf{The approximated WKNN-Shapley preserves the fairness axioms.} 
The Shapley value's axiomatic properties, particularly the \emph{symmetry} and \emph{null player} axioms, are of great importance for ensuring fairness when attributing value to individual players \add{(see Appendix \ref{appendix:related-works} for the formal definitions of the two axioms)}. 
These fundamental axioms have fostered widespread adoption of the Shapley value. 
An ideal approximation of the Shapley value, therefore, should preserve at least the symmetry and null player axioms to ensure that the principal motivations for employing the Shapley value—fairness and equity—are not diminished. The prevalent Monte Carlo-based approximation techniques give randomized results and necessarily muddy the clarity of fairness axioms. In contrast, our deterministic approximation preserves both important axioms. 

\begin{theorem}
\label{thm:approx-shapley-property}
The approximated Shapley value $\{ \phihatmstar \}_{z_i \in D}$ satisfies symmetry and null player axiom. 
\end{theorem}


\begin{remark}[\textbf{Selection of $\mstar$}]
\label{remark:mstar-selection}
\add{Ideally, we would like to pick the smallest $\mstar$ such that $\eps(\mstar)$ is significantly smaller than $|\phi_{z_i}|$ for a significant portion of $z_i$s. However, determining a universally applicable heuristic for setting $\mstar$ is challenging due to the varying magnitude of $\phi_{z_i}$ across different datasets, which are difficult to anticipate. For example, in a case where all but one data point are ``null players'', that single data point will possess a value of $\U(D)$, while all others will be valued at 0. On the other hand, if all data points are identical, each will receive a value of $\U(D)/N$. Therefore, we suggest to select $\mstar$ in an adaptive way. 
Specifically, for each $\mstar \in \{K+1, \ldots, N\}$, we calculate $(\phihatmstar)_{i \in [N]}$, halting the computation when the magnitude of $\eps(\mstar)$ is substantially smaller than (e.g., $< 10\%$) the magnitude of $\phihatmstar$ for a majority of $z_i$s. 
This approach does not increase the overall runtime since $\phihat_{z_i}^{(\mstar+1)}$ can be easily computed from $\phihatmstar$ and the additionally computed $\Fi[\mstar+1, \cdot, \cdot]$. 
Details of this approach are in Appendix \ref{appendix:how-to-select-mstar}. 
Furthermore, we highlight that 
our deterministic approximation algorithm maintains the fairness properties of exact WKNN-Shapley. 
Hence, in practice, we can potentially use a smaller $\mstar$ and still get satisfactory performance
in discerning data quality. 
Throughout our experiments in Section \ref{sec:eval}, we find that setting $\mstar = \sqrt{N}$ consistently works well across all benchmark datasets.}



\end{remark}

\section{Numerical Experiments}
\label{sec:eval}

We systematically evaluate the performance of our WKNN-Shapley computation and approximation algorithms. 
Our experiments aim to demonstrate the following assertions: 
\add{\textbf{(1)} Our exact and deterministic approximation algorithm for WKNN-Shapley significantly improves computational efficiency compared to the baseline $O(N^K)$ algorithm and Monte Carlo approximation, respectively. 
\textbf{(2)} Compared to unweighted KNN-Shapley, WKNN-Shapley \textbf{(2-1)} achieves a better performance in discerning data quality in several important downstream tasks including mislabeled/noisy data detection and data selection, and \textbf{(2-2)} demonstrates more stable performance against different choices of $K$s. \textbf{(3)} The approximated WKNN-Shapley, while being more efficient, consistently achieves performance comparable to the exact WKNN-Shapley on most of the benchmark datasets.
}





\subsection{Runtime Comparison}
\label{sec:eval-runtime}

We empirically assess the computational efficiency of our exact and deterministic approximation algorithms for WKNN-Shapley, comparing them to the $O(N^K)$ exact algorithm and the Monte Carlo approximation presented in \cite{jia2019efficient}. 
We examine various training data sizes $N$ and compare the execution clock time of different algorithms at each $N$. 
In data size regimes where the baseline algorithms from \cite{jia2019efficient} are infeasible to execute ($>10$ hours), we fit a polynomial curve to smaller data size regimes and plot the predicted extrapolation for larger sizes. 

Figure \ref{fig:runtime-nb3} shows that the exact algorithm from \cite{jia2019efficient} requires $\ge 10^3$ hours to run even for $N=100$, rendering it impractical for actual use. 
In contrast, our exact algorithm for computing WKNN-Shapley 
achieves a significantly better computational efficiency (e.g., \textbf{almost $\mathbf{10^6}$ times faster at $\mathbf{N = 10^{5}}$}). 
Our deterministic approximation algorithm is compared to the Monte Carlo approximation from \cite{jia2019efficient}. 
For a fair comparison, both algorithms are aligned for the same theoretical error bounds. 
Note that the error bound for the Monte Carlo algorithm is a high-probability bound, subject to a small failure probability. In contrast, our error bound always holds, offering a more robust guarantee compared to the Monte Carlo technique. 
Nonetheless, from Figure \ref{fig:runtime-nb3} we can see that our deterministic approximation algorithm not only provides a stronger approximation guarantee but also achieves significantly greater efficiency (e.g., also \textbf{almost $\mathbf{10^6}$ times faster at $\mathbf{N = 10^{5}}$}). 
This showcases the remarkable improvements of our techniques. 


\begin{figure}[t]
    \centering
    \centering
    \includegraphics[width=0.5\columnwidth]{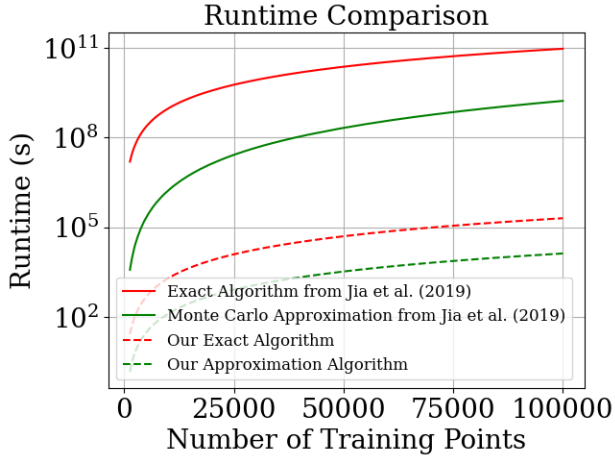}
    \caption{
    Runtime comparison between our exact and approximation algorithms for WKNN-Shapley in Section \ref{sec:shapley-for-binary}, and those from \cite{jia2019efficient}, across varying training data sizes $N$. 
    We set $K = 5$ and the weights are discretized to 3-bit here. 
    \add{In Appendix \ref{appendix:eval}, we provide additional experiments on different $K$s and $b$s. 
    }
    For our deterministic approximation algorithm, we set $\mstar = \sqrt{N}$ (so that the time complexity is $O(N^{1.5})$). 
    For the Monte Carlo approximation from \cite{jia2019efficient}, we align the error bounds to be the same as ours for fair comparison; we set the failure probability for Monte Carlo method as $\delta = 0.1$. 
    The plot shows the average runtime based on 5 independent runs. 
    }
    \label{fig:runtime-nb3}
\end{figure}

\subsection{Discerning Data Quality}
\label{sec:eval-application}

Due to the superior computational efficiency of our newly developed algorithms, WKNN-Shapley has become a practical data valuation technique for actual use. In this section, we evaluate its effectiveness in discerning data quality for common real-world applications. 
\textbf{Tasks:} we consider three applications that are commonly used for evaluating the performance of data valuation techniques in the prior works \cite{ghorbani2019data, kwon2022beta, wang2023data}: mislabeled data detection, noisy data detection, and data selection. 
\add{Due to space constraints, we only present the results for mislabeled data detection here, and defer the results for the other two tasks to Appendix \ref{appendix:eval}. Mislabeled data usually detrimentally impacts model performance. Hence, a reasonable data valuation technique should assign low values to these data points. In our experiments for mislabeled data detection, we randomly select 10\% of the data points to flip their labels.}
\textbf{Baselines \& Settings \& Hyperparameters:}
We evaluate the performance of both the exact and approximated WKNN-Shapley. We use $\ell_2$ distance and the popular RBF kernel $w_i = \exp(-\norm{x_i - \xval})$ to determine the weights of training points. 
\add{We discretize the weights to 3 bits, as we find this level of precision offers a balance between good performance and computational efficiency, with a weight space size, $\nwspace$, of merely $2^3 = 8$. In Appendix \ref{appendix:eval}, we conduct ablation studies on the choice of the number of bits for discretization. For approximated WKNN-Shapley, we set $\mstar = \sqrt{N}$.} 
Our primary baseline is the unweighted, soft-label KNN-Shapley from \cite{jia2019efficient}. 
\add{Since our WKNN-Shapley corresponds to hard-label KNN, we also include unweighted, hard-label WKNN-Shapley in comparison for completeness. Note that it can be computed by simply setting the weights of all data points as a constant.}


\textbf{Results.} 
We use AUROC as the performance metric on mislabeled data detection tasks. 
\textbf{Unweighted vs weighted KNN-Shapley:} 
Table \ref{tb:mislabel-detection} shows the AUROC scores across the 13 benchmark datasets we experimented on when $K=5$. 
Notably, both exact and approximated WKNN-Shapley markedly outperform the unweighted KNN-Shapley (either soft-label or hard-label) across most datasets. This can likely be attributed to WKNN-Shapley’s ability to more accurately differentiate between bad and good data based on the proximity to the queried example. In Appendix \ref{appendix:eval-comparison}, we present a qualitative study highlighting why WKNN-Shapley outperforms unweighted KNN-Shapley in discerning data quality. 
\textbf{Exact vs Approximated WKNN-Shapley:} 
From Table \ref{tb:mislabel-detection}, we can see an encouraging result that the approximated WKNN-Shapley achieves performance comparable to \add{(and sometimes even slightly better than)} the exact WKNN-Shapley across the majority of datasets. This is likely attributable to its favored property in preserving the fairness properties of its exact counterpart. 
\textbf{Robustness to the choice of $K$:} 
In Figure \ref{fig:ablation-varyK}, we show that, compared to unweighted KNN-Shapley, WKNN-Shapley maintains notably stable performance across various choices of $K$, particularly for larger values. 
This is because those benign data points—though within the $K$ nearest neighbors of the query example and possessing different labels—may not receive a very low value due to their likely distant positioning from the query example. Conversely, unweighted KNN-Shapley tends to assign these benign points lower values.

\begin{table}[t]
\centering
\setlength\abovecaptionskip{2pt}
\setlength\belowcaptionskip{-10pt}
\resizebox{0.8\columnwidth}{!}{\input{mislabel_detect_table}}
\caption{AUROC scores of different variants of KNN-Shapley for mislabeled data detection on benchmark datasets. The higher, the better. 
}
\label{tb:mislabel-detection}
\end{table}

\begin{figure}[t]
    \centering
    \setlength\abovecaptionskip{0pt}
    \setlength\belowcaptionskip{-10pt}
    \centering
    \includegraphics[width=0.8\columnwidth]{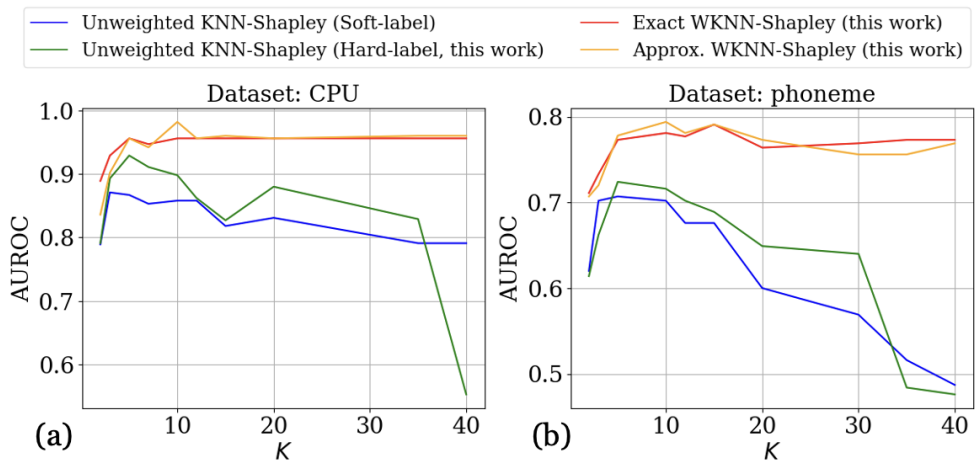}
    \caption{
    AUROC scores of different variants of KNN-Shapley for mislabeled data detection with different $K$s. The higher the curve is, the better the method is. 
    }
    \label{fig:ablation-varyK}
\end{figure}

\vspace{-2mm}
\section{Conclusion}
\label{sec:conclusion}

\vspace{-2mm}
In this study, we addressed WKNN-Shapley computation and approximation when using the accuracy of hard-label KNN with discretized weights as the utility function. 
Future work should explore polynomial time computation of exact Data Shapley for other learning algorithms. It is especially important to think about whether we can make some modifications to more complicated learning algorithms, e.g., neural networks, so that their exact Data Shapley can be computed efficiently. 

\input{acknowledgements}

\newpage

\bibliographystyle{apalike}
\bibliography{ref}


\newpage
\onecolumn

\inappendixtrue
\appendix

\section{Extended Related Works}
\label{appendix:related-works}
\input{appendix_relatedworks}

\newpage

\section{Additional Background of KNN-Shapley}
\label{appendix:knn-background}
\input{appendix_knnsv_background}

\newpage

\section{Additional Discussion of Weights Discretization}
\label{appendix:error-disc}
\input{appendix_errordisc}

\newpage

\section{Additional Details for Exact and Approximation Algorithm for WKNN-Shapley}
\label{appendix:method-details}

We state the full version of Theorem \ref{thm:Fi-recursive} here, and defer the proof to Appendix \ref{appendix:proof}. 

\begin{theorem}[Full version of Theorem \ref{thm:Fi-recursive}]
\label{thm:Fi-recursive-full}
When $K > 1$,\footnote{Since \cite{jia2019efficient} has shown that weighted KNN-Shapley can be computed in $O(N^K)$ time complexity, we focus on the setting where $K > 1$.} for $\ell = 1$ we have
$
\Fi[m, 1, s] = 
\begin{cases} 
1 & s = w_m \\
0 & s \ne w_m
\end{cases}
$. 
We can then compute $\Fi[m, \ell, s]$ for $\ell \ge 2$ with the following relations: 

If $\ell \le K-1$, we have
\begin{align}
    \Fi[m, \ell, s] = \sum_{t=1}^{m-1} \Fi[t, \ell-1, s - w_m]
\label{eq:recursive}
\end{align}
and if $\ell \ge K$, we have 
\begin{align}
    \Fi[m, \ell, s] = 
    \begin{cases}
        0 & m < i \\
        \sum_{t=1}^{m-1} \Fi[t, K-1, s] {N-m \choose \ell-K} & m > i
    \end{cases}
\label{eq:explicit}
\end{align}
Note that we set $\Fi[i, \cdot, \cdot] = 0$ for mathematical convenience. 
\end{theorem}

In the following, we present the pseudocode for our exact computation and approximation algorithm for WKNN-Shapley. 
For the clarity of presentation, we first show a reader-friendly version but inefficient version of the pseudo-code for the exact computation algorithm from Section \ref{sec:exact-shapley} in Appendix \ref{appendix:pseudocode-readable}. 
We then show the pseudo-code that optimizes the runtime (but less readable) in Appendix \ref{appendix:pseudocode}. 
We further show the pseudo-code for our deterministic approximation algorithm in Appendix \ref{appendix:pseudocode-approx}. 

\textbf{Notation.} 
Recall that we use $\wspace$ to denote the discretized space of $[0, 1]$, where we create $2^b$ equally spaced points within the interval when we use $b$ bits for discretization. We denote $\nwspace := |\wspace| = 2^b$ the size of the weight space. Furthermore, we use $\wspace_{(K)}$ to denote the discretized space of $[0, K]$ (where we create $K 2^b$ equally spaced points within the interval).

\newpage

\subsection{Reader-friendly Pseudo-code}
\label{appendix:pseudocode-readable}

Here, we show a reader-friendly version of the pseudo-code for our algorithms for computing WKNN-Shapley for binary classification setting. 
The runtime-optimized version of the pseudo-code is shown in Appendix \ref{appendix:pseudocode}. 

\inappendixfalse
\input{newpseudocode-readable}

\newpage

\subsection{Detailed Pseudo-code used in Implementation}
\label{appendix:pseudocode}

Here, we show the runtime-optimized version of the pseudo-code for our algorithms for computing WKNN-Shapley for binary classification setting. 
Specifically, the for-loops for computing $\Fi$ and $\Rim$ can be optimized for efficiency. 

\inappendixtrue

\input{newpseudocode-readable}
\newpage

\subsection{Detailed Pseudo-code for approximation algorithm}
\label{appendix:pseudocode-approx}

Here, we show the pseudo-code for our deterministic approximation algorithms for WKNN-Shapley from Section \ref{sec:deterministic-approx}. We highlight its difference with the exact computation algorithm. 

\input{pseudocode-approx}

\newpage

\subsection{Expanded Discussion for the Selection of $\mstar$ for Deterministic Approximation Algorithm}
\label{appendix:how-to-select-mstar}

\textbf{Determining a universally applicable heuristic for setting $\mstar$ is challenging.} 
As we mentioned in Remark \ref{remark:mstar-selection}, ideally, we would like to pick the smallest $\mstar$ such that the error bound $\eps(\mstar)$ from Theorem \ref{thm:error-bound} is significantly smaller than $|\phi_{z_i}|$ for a significant portion of $z_i$s. However, determining a universally applicable heuristic for setting $\mstar$ is challenging due to the varying magnitude of $\phi_{z_i}$ across different datasets, which are difficult to anticipate. For example, in the case where all but one data point are ``null players'', that single data point will possess a value of $\U(D)$, while all others will be valued at 0. On the other hand, if all data points are identical, each will receive a value of $\U(D)/N$. 

\textbf{We can select $\mstar$ in an adaptive way (not used in our experiment).} 
As the magnitude of WKNN-Shapley can vary significantly depending on the specific dataset, it is more reasonable to select $\mstar$ in an adaptive way. 
Specifically, we can compute $(\phihatmstar)_{i \in [N]}$ for each of $\mstar = K+1, \ldots,$. We can keep increment $\mstar$ until the magnitude of $\eps(\mstar)$ is substantially smaller than (e.g., $< 10\%$) the magnitude of $\phihatmstar$ for a majority of $z_i$s. 
We can easily modify Algorithm \ref{alg:efficient-wknn-approx} so that this approach does not increase the overall runtime since $\phihatmstar$ can be easily computed from $\phihat_{z_i}^{(\mstar-1)}$ and the additionally computed $\Fi[\mstar, \cdot, \cdot]$. 
That is, from Definition \ref{def:approxmstar}, we can easily see that when $y_i = \yval$, 
\begin{align}
\phihatmstar
= 
\phihat_{z_i}^{(\mstar-1)} + 
\left(
\frac{1}{N} \sum_{\ell=0}^{K-1} \frac{ \sum_{s \in [-\wtil_i, 0)} \Fi\left[\mstar, \ell, s\right] }{{N-1 \choose \ell}} + 
\frac{ \texttt{R}_{i, \mstar} }{ \mstar {\mstar-1 \choose K} } \right)
\end{align}
and when $y_i \ne \yval$, 
\begin{align}
\phihatmstar
= 
\phihat_{z_i}^{(\mstar-1)} - 
\left( 
\frac{1}{N} \sum_{\ell=0}^{K-1} \frac{ \sum_{s \in [0, -\wtil_i)} \Fi\left[\mstar, \ell, s\right] }{{N-1 \choose \ell}} + 
\frac{ \texttt{R}_{i, \mstar} }{ \mstar {\mstar-1 \choose K} }
\right)
\end{align}

\textbf{In our experiment, we find $\mstar = \sqrt{N}$ works well across all benchmark datasets (ablation study in Appendix \ref{appendix:eval-varyEPS}).} 
However, in our experiment, we find that the performance of approximated WKNN-Shapley on the downstream tasks such as mislabeled and noisy data detection, are relatively stable across a wide range of choices of $\mstar$ (see Appendix \ref{appendix:eval-varyEPS}). 
This is likely because of the nice property of our approximated WKNN-Shapley in preserving the fairness properties of the exact WKNN-Shapley (Theorem \ref{thm:approx-shapley-property}). 
Hence, in our experiment, we do not follow the adaptive process of selecting $\mstar$, but adopt the simple rule of $\mstar = \sqrt{N}$, which already works very well. 

\newpage

\section{Extension to Multi-class Classification Setting}
\label{appendix:multiclass}
\input{appendix_multiclass}

\newpage

\section{Missing Proofs}
\label{appendix:proof}
\input{appendix_proofs}

\newpage

\section{Evaluation Settings \& Additional Experiments}
\label{appendix:eval}

\input{appendix_evalsetting}

\input{appendix_evalexp}

\end{document}

%% file: macro.tex
\newif\iffinal
\finaltrue

\iffinal
    \newcommand{\tianhao}[1]{}
    \newcommand{\ruoxi}[1]{}
    \newcommand{\add}[1]{#1}
\else
    \newcommand{\tianhao}[1]{{\bf \textcolor{purple}{[Tianhao: #1]}}}
    \newcommand{\ruoxi}[1]{{\bf \textcolor{BrickRed}{[Ruoxi:#1]}}}
    \newcommand{\add}[1]{\textcolor{red}{#1}}
\fi

%% file: math_command.tex
\usepackage{cleveref}
\usepackage{bbm}

\newtheorem{theorem}{Theorem}
\newtheorem{lemma}[theorem]{Lemma}

\newtheorem{definition}[theorem]{Definition}

\newtheorem{remark}{Remark}

\newtheorem{remark-star}{Remark}
\newtheorem{remark-star-1}{Remark}

\newtheorem*{proof-sketch}{Proof Sketch}
\usepackage{booktabs}

\DeclareMathOperator*{\argmax}{\arg\!\max}

\newcommand{\R}{\mathbb{R}}
\newcommand{\eps}{\varepsilon}

\newcommand{\norm}[1]{\left\lVert#1\right\rVert}

\newcommand{\ind}{\mathbbm{1}}

\newcommand{\A}{\mathcal{A}}

\newcommand{\U}{v}
\newcommand{\test}{\mathrm{val}}

\newcommand{\metric}{\texttt{ValAcc}}

\newcommand{\NBK}{ \texttt{NB}_{x^{(\test)}, K} }

\newcommand{\Dval}{ D^{(\test)} }

\newcommand{\red}[1]{\textcolor[RGB]{255, 49, 49}{#1}}

\newcommand{\cmt}[1]{\texttt{// #1}}

\newcommand{\zval}{ z^{(\test)} }
\newcommand{\xval}{ x^{(\test)} }
\newcommand{\yval}{ y^{(\test)} }

\newcommand{\C}{\mathbf{C}}

\newcommand{\weightfunc}{\omega_{\xval}}

\usepackage{tikz}

\newcommand*\circled[1]{\tikz[baseline=(char.base)]{
            \node[shape=circle,draw,inner sep=1pt] (char) {#1};}}

\newcommand{\Gil}{\texttt{G}_{i, \ell}}
\newcommand{\Gizero}{\texttt{G}_{i, 0}}

\newcommand{\Giltil}{\widetilde{\texttt{G}}_{i, \ell}}

\newcommand{\Giltilmstar}{\widetilde{\texttt{G}}_{i, \ell}^{(\mstar)}}

\newcommand{\Fi}{\texttt{F}_i}

\newcommand{\Fihat}{\widehat{\texttt{F}}_i}

\newcommand{\Rim}{\texttt{R}_{i, m}}

\newcommand{\Rimhat}{\widehat{\texttt{R}}_{i, m}}

\newcommand{\wspace}{\mathbf{W}}
\newcommand{\nwspace}{W}

\newcommand{\s}{\mathbf{s}}
\newcommand{\e}{\mathbf{e}}

\newcommand{\mstar}{M^\star}
\newcommand{\phihat}{\widehat \phi}
\newcommand{\phihatmstar}{\phihat_{z_i}^{(\mstar)}}

\newcommand{\sign}{ \texttt{sign} }

\renewcommand{\cmt}[1]{\texttt{// #1}}

\newcommand{\wtil}{\widetilde{w}}

\newcommand{\Ualter}{\widetilde{v}}

\newcommand{\condKNN}{\texttt{Cond}_{K\text{NN}}}
\newcommand{\condNegToPos}{\texttt{Cond}_{\text{0to1}}}
\newcommand{\condPosToNeg}{\texttt{Cond}_{\text{1to0}}}

\newif\ifinappendix

%% file: mislabel_detect_table.tex
\begin{tabular}{@{}ccccc@{}}
\toprule
& \textbf{\begin{tabular}[c]{@{}c@{}}Unweighted \\ KNN-Shapley \\ (Soft-label) \end{tabular}} & \textbf{\begin{tabular}[c]{@{}c@{}}Unweighted \\ KNN-Shapley \\ (Hard-label, \add{this work})\end{tabular}} & \textbf{\begin{tabular}[c]{@{}c@{}}Exact\\ WKNN-Shapley\\ \add{(this work)}\end{tabular}} & \textbf{\begin{tabular}[c]{@{}c@{}}Approximated\\ WKNN-Shapley\\ \add{(this work)}\end{tabular}} \\ \midrule
\textbf{2DPlanes}   & 0.849                                                                                      & 0.8                                                                                        & \textbf{0.884}                                                                              & 0.831                                                                                     \\
\textbf{CPU}        & 0.867                                                                                      & 0.929                                                                                      & \textbf{0.956}                                                                              & \textbf{0.956}                                                                                     \\
\textbf{Phoneme}    & 0.707                                                                                      & 0.724                                                                                      & 0.773                                                                              & \textbf{0.778}                                                                                     \\
\textbf{Fraud}      & 0.556                                                                                      & 0.547                                                                                      & \textbf{0.751}                                                                              & 0.596                                                                                     \\
\textbf{Creditcard} & 0.698                                                                                      & 0.676                                                                                      & \textbf{0.842}                                                                              & 0.747                                                                                     \\
\textbf{Vehicle}    & 0.689                                                                                      & 0.724                                                                                      & 0.8                                                                                & \textbf{0.813}                                                                                     \\
\textbf{Click}      & 0.627                                                                                      & 0.6                                                                                        & \textbf{0.751}                                                                              & 0.693                                                                                     \\
\textbf{Wind}       & 0.836                                                                                      & 0.849                                                                                      & 0.858                                                                              & \textbf{0.88}                                                                                      \\
\textbf{Pol}        & 0.907                                                                                      & 0.862                                                                                      & \textbf{1}                                                                                  & 0.991                                                                                     \\
\textbf{MNIST}      & 0.724                                                                                      & 0.471                                                                                      & 0.831                                                                              & \textbf{0.836}                                                                                     \\
\textbf{CIFAR10}    & 0.684                                                                                      & \textbf{0.76}                                                                                       & \textbf{0.76}                                                                               & 0.756                                                                                     \\ 
\textbf{AGNews}     & 0.953                                                                                            & 0.978                                                                                     & \textbf{0.991}                                                                               & 0.988                                                                                      \\
\textbf{DBPedia}    & 0.968                                                                                            & 0.902                                                                                     & \textbf{1}                                                                                   & \textbf{1}                                                                                          \\ \bottomrule
\end{tabular}

%% file: acknowledgements.tex
\section*{Acknowledgments}
This work was supported in part by the National Science Foundation under grants CNS-2131938, CNS-1553437, CNS-1704105, IIS-2312794, IIS-2313130, OAC-2239622, the ARL’s Army Artificial Intelligence Innovation Institute (A2I2), the Office of Naval Research Young Investigator Award, the Army Research Office Young Investigator Prize, Schmidt DataX award, Princeton E-ffiliates Award, \add{Amazon-Virginia Tech Initiative in Efficient and Robust Machine Learning, the Commonwealth Cyber Initiative,} and a Princeton's Gordon Y. S. Wu Fellowship. We are grateful to anonymous reviewers at AISTATS for their valuable feedback.

%% file: appendix_relatedworks.tex
\subsection{Data Shapley} 

\emph{Data Shapley} is one of the first principled approaches to data valuation that became increasingly popular \cite{ghorbani2019data, jia2019towards}. 
Data Shapley is based on the \emph{Shapley value}, a famous solution concept from game theory literature which is usually justified as the \emph{unique} value notion satisfying the following four axioms: 
\begin{enumerate}[label=\textbf{(\arabic*)}]
    \item \textbf{Null player:} if $\U(S \cup \{z_i\})=\U(S)$ for all $S \subseteq D \setminus \{z_i\}$, then $\phi_{z_i}(\U)=0$. 
    \item \textbf{Symmetry:} if $\U(S \cup \{z_i\}) = \U(S \cup \{z_j\})$ for all $S \subseteq D \setminus \{z_i, z_j\}$, then $\phi_{z_i}(\U)=\phi_{z_j}(\U)$. 
    \item \textbf{Linearity:} For utility functions $\U_1, \U_2$ and any $\alpha_1, \alpha_2 \in \R$, $\phi_{z_i}(\alpha_{1} \U_{1}+\alpha_{2} \U_{2}) = \alpha_{1} \phi_{z_i}(\U_{1}) + \alpha_{2} \phi_{z_i}(\U_{2})$. 
    \item \textbf{Efficiency:} for every $\U, \sum_{z_i \in D} \phi_{z_i}(\U)=\U(D)$.
\end{enumerate}

Since its introduction, numerous variants of Data Shapley have been developed \cite{jia2019efficient, ghorbani2020distributional, wang2020principled, bian2021energy, kwon2022beta, lin2022measuring, wu2022davinz, karlavs2022data, wang2023data, wang2023notegroup,liu20232d, wang2023threshold}, reflecting its effectiveness as a principled approach for quantifying data point contributions to ML model training. 
However, not all axioms mentioned above are considered necessary for a reasonable data valuation technique by the community. 
For example, \cite{kwon2022beta} argue that \textbf{(4) Efficiency} is not necessarily required for data valuation, and \cite{yan2020ifyoulike} argue that \textbf{(3) Linearity} is mainly a technical requirement and does not have a natural interpretation in the context of machine learning. 
That said, the \textbf{(1) Null player} and \textbf{(2) Symmetry} axioms are often seen as the ``fairness axioms'' and are deemed fundamental for any sound data valuation technique. Therefore, a reasonable Shapley approximation should uphold these two axioms to ensure the Shapley value's core advantages—fairness and equity—are not diminished. 
The deterministic approximation algorithm we develop in Section \ref{sec:deterministic-approx} meets these criteria.


\subsection{KNN-Shapley} 
The computation of the Shapley value is notoriously resource-intensive. 
To the best of our knowledge, unweighted KNN stands as the \emph{only} frequently deployed ML model where the exact Data Shapley can be efficiently computed. 
Due to its exceptional computational efficiency coupled with its capability to discern data quality, KNN-Shapley has become one of the most popular and practical data valuation techniques. 
For instance, \cite{ghorbani2022data} extends KNN-Shapley to active learning, \cite{shim2021online} applies it in a continual learning setting. 
Additionally, studies such as \cite{liang2020beyond,liang2021herald} have leveraged KNN-Shapley to eliminate ambiguous samples in NLP tasks, and \cite{courtnage2021shapley} has endorsed its use in semi-supervised learning data valuation. 
\cite{belaid2023optimizing} extends the analysis of KNN-Shapley to the calculation to the Shapley interaction index. 
KNN-Shapley has also shown its practicality in real-world scenarios. For example, \cite{pandl2021trustworthy} shows that KNN-Shapley is the \emph{only} practical data valuation technique for valuing large amounts of healthcare data. 
A concurrent work \cite{wang2023threshold} considers a simple variant of unweighted KNN termed \emph{Threshold KNN}, and develops an alternative of KNN-Shapley that can be easily incorporated with differential privacy. 
Both \cite{wang2023threshold} and our work demonstrate the importance of proper adjustments to the underlying KNN's configuration (i.e., the utility function) in developing new data valuation techniques with desired properties (computational efficiency, privacy compatibility, etc.). 

All the studies mentioned above focus on \emph{unweighted} KNN. 
On the other hand, \emph{weighted} KNN incorporates more information about the underlying dataset. 
Consequently, the Data Shapley score derived from weighted KNN might offer a better assessment of individual data point quality (as demonstrated in our experiments).

Finally, we note that some alternative data valuation techniques are as efficient as KNN-Shapley \cite{just2022lava, kwon2023data}. However, these methods lack a formal theoretical justification as the Shapley value-based approaches.


%% file: appendix_knnsv_background.tex
\textbf{Notation Review.} 
Recall that we use $\alpha_{\xval}^{(S, j)}$ denotes the index (among $D$) of $j$th closest data point in $S$ to $\xval$.

\subsection{Unweighted KNN-Shapley}
\label{appendix:unweighted-knnsv}

The unweighted KNN-Shapley was originally proposed in \cite{jia2019efficient} and was later refined in \cite{wang2023noteknn}. 
Specifically, \cite{wang2023noteknn} considers the utility function for unweighted, soft-label KNN on a validation point $\zval$: 
\begin{align}
\U(S; \zval) &:= \frac{ \sum_{j=1}^{\min(K, |S|)} \ind[y_{\alpha_{\xval}^{(S, j)}} = \yval] }{\min(|S|, K)}
\label{eq:util-unweightedknn}
\end{align}
which is a special case of the utility function for weighted KNN in (\ref{eq:util-weighted-softlabel}). 
It can be interpreted as the probability of a soft-label KNN classifier in predicting the correct label for a validation point $\zval = (x^{(\test)}, y^{(\test)}) \in \Dval$. 
When $|S|=0$, $\U_{\zval}(S)$ is set to the accuracy by random guessing (i.e., 
$\U(\emptyset) = 1/\texttt{\#class}$). 
Here is the main result of \cite{wang2023noteknn}:

\begin{theorem}[KNN-Shapley \cite{wang2023noteknn}]
\label{thm:knn-shapley-full}
Consider the utility function in (\ref{eq:util-unweightedknn}). 
Given a validation data point $\zval = (x^{(\test)}, y^{(\test)})$ and a distance metric $d(\cdot, \cdot)$, if we sort the training set $D = \{z_i = (x_i, y_i)\}_{i=1}^N$ according to $d(x_i, x^{(\test)})$ in ascending order, then the Shapley value of each data point $\phi_{z_i}$ corresponding to utility function $\U_{\zval}$ can be computed recursively as follows:
\begin{align*}
    \phi_{z_N}  &= \frac{\ind[N \ge 2]}{N} \left( \ind[y_N = y^{(\test)}] - \frac{ \sum_{i=1}^{N-1} \ind[y_i = y^{(\test)}] }{N-1} \right) \left( \sum_{j=1}^{ \min(K, N) - 1 } \frac{1}{j+1} \right) + \frac{1}{N} \left( \ind[y_N = y^{(\test)}] - \frac{1}{C} \right) \\
    \phi_{z_i}  &= \phi_{z_{i+1}} + 
    \frac{ \ind[y_i = y^{(\test)}] - \ind[y_{i+1} = y^{(\test)}] }{N-1} 
    \left[
    \sum_{j=1}^{\min(K, N)} \frac{1}{j} 
    + \frac{\ind[N \ge K]}{K} \left( \frac{ \min(i, K) \cdot (N-1) }{i} - K
    \right)
    \right]
\end{align*}
where $C$ denotes the number of classes for the classification task. 
\end{theorem}

The computation of all unweighted KNN-Shapley $(\phi_{z_1}, \ldots, \phi_{z_N})$ can be achieved in $O(N \log N)$ runtime in total, as the runtime is dominated by the sorting data points in $D$.

\subsection{Baseline Algorithm for Computing and Approximating WKNN-Shapley}
\label{appendix:baseline}

\paragraph{Exact Computation Algorithm.}
\cite{jia2019efficient} shows that the Data Shapley for weighted KNN can be computed exactly with a runtime of $O(N^K)$. 
The high-level idea, as described in Section \ref{sec:baseline-algorithm} in the maintext, is that we only need to consider evaluating $\U(S)$ for those $S$ where the addition of the target data point $z_i$ may change the prediction of KNN. Moreover, there are at most $O(N^K)$ such $S$s. 
For completeness, we state the specific expression for computing the exact WKNN-Shapley from \cite{jia2019efficient}. 
We note that the original theorem statement in \cite{jia2019efficient} has minor errors and we also fix it here. 

\begin{theorem}[\cite{jia2019efficient}]
\label{thm:WKNN-from-Jia}
Consider the utility function in (\ref{eq:util-weighted-softlabel}). 
Let $B_k(i)=\{S:|S|=k, z_i \notin S, S \subseteq D\}$. 
Let $r(\cdot)$ be a function that maps the set of training data to their ranks of similarity to $\xval$. Then, the Shapley value $\phi_{z_i}$ of each training point $z_i$ can be calculated recursively as follows:
\begin{align}
&\phi_{z_N} = \frac{1}{N}\sum_{k=0}^{K-1} \frac{1}{{N-1\choose k}} \!\!\sum_{S\in B_k(z_N)}\!\!
  \left[\U(S\cup \{z_N\}) - \U(S)\right]\\
&\phi_{z_{i}} = \phi_{z_{i+1}} + \frac{1}{N-1}\sum_{k=0}^{N-2} \frac{1}{{N-2\choose k}} \sum_{S\in D_{i,k}} A_{i,k} \left[\U(S\cup \{z_i\}) - \U(S\cup \{z_{i+1}\}) \right]
\end{align}
where 
\begin{align*}
D_{i,k} = 
\begin{cases}
B_k(z_i) \cap B_k(z_{i+1}) & 0\leq k\leq K-2 \\
B_{K-1}(z_i)\cap B_{K-1}(z_{i+1}) & K-1\leq k\leq N-2
\end{cases}
\end{align*}
and 
\begin{align*}
A_{i,k} = 
\begin{cases}
    1 & 0\leq k\leq K-2 \\
    {N - \max \left(i+1, \alpha_{\xval}^{(S, |S|)} \right) \choose k-K+1} & K-1\leq k\leq N-2
\end{cases}
\end{align*}
\end{theorem}
\begin{proof}
The proof proceeds by standard combinatorial analysis and we refer readers to Appendix E.2. in \cite{jia2019efficient} for details. We note that $\alpha_{\xval}^{(S, |S|)}$ means the index of the farthest data point to $\xval$ in $S$.  
\end{proof}

\begin{remark}
While \cite{jia2019efficient} state the above results for soft-label weighted KNN, this algorithm is in fact fairly general and also applies to hard-label weighted KNN. 
\end{remark}

\paragraph{Monte Carlo Algorithm.} 
The Monte Carlo approximation algorithm for WKNN-Shapley from \cite{jia2019efficient} is a simple adaptation of the famous \emph{permutation sampling} algorithm \cite{maleki2015addressing}. We refer the reader to Section 5 of \cite{jia2019efficient} for the detailed description and error analysis.



%% file: appendix_errordisc.tex
\paragraph{Value Rank Preservation.} We show that the Shapley value computed based on the discrete weights has the same ranking order compared with the Shapley value computed on the continuous weights (it might create ties but will not reverse the original order).

Here, we consider the binary classification setting since our approach for computing WKNN-Shapley for multi-class classification (in Appendix \ref{appendix:multiclass}) makes a reduction to the binary classification setting. 
We make a very mild assumption that if $d(x_i, \xval) \le d(x_j, \xval)$, then $w_i \ge w_j$. That is, the closer the training point is to the query, the higher the weight the training point will have. 
This is natural for any reasonable weighting scheme for KNN. 

\newcommand{\Udisc}{\U^{\texttt{disc}}}

\begin{lemma}
Under the assumption of weight function stated above, for any weight assignment function and any two data points $z_i, z_j$, 
if $\phi_{z_i}(\U) \ge \phi_{z_j}(\U)$, then we have $\phi_{z_i}(\Udisc) \ge \phi_{z_j}(\Udisc)$ where $\Udisc$ refers to the new utility function after weights discretization. 
\end{lemma}
\begin{proof}
The following cases will result in $\phi_{z_i}(\U) \ge \phi_{z_j}(\U)$: 

\textbf{Case 1:} if $y_i = \yval$, $y_j \ne \yval$, it is easy to see that for any $S \subseteq D$ we have $\U(S \cup z_i) \ge \U(S)$ while $\U(S \cup z_j) \le \U(S)$, and hence we have $\phi_{z_i} \ge 0$ while $\phi_{z_i} \le 0$, regardless of the specific magnitude of the weights. Hence, we still have $\phi_{z_i}(\Udisc) \ge 0$ and $\phi_{z_j}(\Udisc) \le 0$ after weights discretization. 

\textbf{Case 2:} if $y_i = y_j = \yval$, then $\phi_{z_i}(\U) \ge \phi_{z_j}(\U)$ implies that $d(x_i, \xval) \le d(x_j, \xval)$ and $w_i \ge w_j$. This is because for any $S \subseteq D \setminus \{z_i, z_j\}$, if $\U(S \cup \{z_j\}) - \U(S) = 1$, we must also have $\U(S \cup \{z_i\}) - \U(S) = 1$. Since weight discretization does not change the rank order of weights $w_i$ and $w_j$, we still have $\phi_{z_i}(\Udisc) \ge \phi_{z_j}(\Udisc)$. 

\textbf{Case 3:} if $y_i = y_j \ne \yval$, then $\phi_{z_i}(\U) \ge \phi_{z_j}(\U)$ implies that $d(x_i, \xval) \ge d(x_j, \xval)$ and $w_i \le w_j$. This is because for any $S \subseteq D \setminus \{z_i, z_j\}$, if $\U(S \cup \{z_i\}) - \U(S) = -1$, we must also have $\U(S \cup \{z_j\}) - \U(S) = 1$. Since weight discretization does not change the rank order of weights $w_i$ and $w_j$, we still have $\phi_{z_i}(\Udisc) \ge \phi_{z_j}(\Udisc)$. 
\end{proof}

\paragraph{Value Deviation and Performance on Downstream Tasks.} 
It is difficult to analytically derive or upper bound the deviation of WKNN-Shapley score caused by weight discretization. Therefore, in Appendix \ref{sec:eval-discretization}, we empirically investigate such value deviation. We also evaluate the impact of weight discretization on WKNN-Shapley's performance on downstream tasks such as mislabeled/noisy data detection. Overall, we conclude that the impact from weight discretization is very small even when the number of discretization bits $b$ is as small as $3$.

%% file: newpseudocode-readable.tex
\centerline{
\scalebox{0.8}{
\begin{minipage}[t]{1\linewidth}
\begin{algorithm}[H]
\ifinappendix
    \caption{Weighted KNN-Shapley for binary classification}
    \label{alg:efficient-wknn}
\else
    \caption{Weighted KNN-Shapley for binary classification (reader-friendly version)}
    \label{alg:efficient-wknn-readable}
\fi
\begin{algorithmic}[1]

\State \textbf{Input:} 
\begin{itemize}
    \item $K$ -- hyperparameter of weighted KNN algorithm. 
    \item $\zval = (\xval, \yval)$ -- the validation point. 
    \item $D = \{z_i = (x_i, y_i)\}_{i=1}^N$ -- sorted training set where $d(x_i, \xval) \le d(x_j, \xval)$ for any $i \le j$.
\end{itemize}

\State

\State Compute the weight $w_i$ for $i \in \{1, \ldots, N\}$. 

\State $\widetilde{w}_{j} = (2 \ind[\yval = y_{j}] - 1) w_{j}$ for $i \in \{1, \ldots, N\}$.

\State

\For{$i \in \{1, \ldots, N\}$}

\State

\State \cmt{Initialize $\Fi$}
\State Initialize $\Fi[m, \ell, s] = 0$ for $m \in \{1, \ldots, N \}, \ell \in \{1, \ldots, K-1\}, s \in \wspace_{(K)}$. 

\For{$m \in \{ 1, \ldots, N \} \setminus \{i\}$}
\State $\Fi[m, 1, \wtil_m] = 1$ \texttt{\textcolor{blue}{(Theorem \ref{thm:Fi-recursive-full})}}
\EndFor

\State 

\ifinappendix
    \State \cmt{Compute $\Fi$ \textcolor{blue}{(Runtime-optimized version)}}
    \For{$\ell \in \{2, \ldots, K-1\}$}
    \State $F_0[:] = \sum_{t=1}^{\ell-1} \Fi[t, \ell-1, :]$
    \For{$m \in \{ \ell, \ldots, N \} \setminus \{i\}$}
    \For{$s \in \wspace_{(K)}$}
    \State $\Fi[m, \ell, s] = F_0[s - w_m]$
    \EndFor
    \EndFor
    \EndFor
\else
    \State \cmt{Compute $\Fi$ \textcolor{blue}{(Runtime-optimized version is in Appendix \ref{appendix:pseudocode})}}
    
    \For{$\ell \in \{2, \ldots, K-1\}$}
    \For{$m \in \{ \ell, \ldots, N \} \setminus \{i\}$}
    \For{$s \in \wspace_{(K)}$}
    \State $\Fi[m, \ell, s]
    = \sum_{t=1}^{m-1} \Fi[t, \ell-1, s - \wtil_m]$ 
    \texttt{\textcolor{blue}{(Equation (\ref{eq:recursive}) in Theorem \ref{thm:Fi-recursive-full})}}
    \EndFor
    \EndFor
    \EndFor
\fi

\State

\ifinappendix
    \State \cmt{Compute $\Rim$ \textcolor{blue}{(Runtime-optimized version)}}
    \For{$s \in \wspace_{(K)}$}
    \State 
    $R_0[s] = \sum_{t=1, t \ne i}^{ \max(i+1, K+1)-1 } \Fi[t, K-1, s]$. 
    \EndFor
    \For{$m \in \{ \max(i+1, K+1), \ldots, N \}$}
    \State 
    $\Rim = 
    \begin{cases}
    \sum_{s \in [-\wtil_i, -\wtil_m)} R_0[s] & \text{for } y_i = \yval \\
    \sum_{s \in [-\wtil_m, -\wtil_i)} R_0[s] & \text{for } y_i \ne \yval 
    \end{cases}
    $
    \State $R_0 = R_0 + \Fi[m, K-1, :]$
    \EndFor
\else
    \State \cmt{Compute $\Rim$ \textcolor{blue}{ (Runtime-optimized version is in Appendix \ref{appendix:pseudocode})}}
    
    \For{$m \in \{ \max(i+1, K+1), \ldots, N \}$}
    \State 
    $\Rim = 
    \begin{cases}
    \sum_{t=1}^{m-1} \sum_{s \in [-\wtil_i, -\wtil_m)} \Fi[t, K-1, s] & \text{for } y_i = \yval \\
    \sum_{t=1}^{m-1} \sum_{s \in [-\wtil_m, -\wtil_i)} \Fi[t, K-1, s] & \text{for } y_i \ne  \yval 
    \end{cases}
    $ \textcolor{blue}{ \texttt{(Theorem \ref{thm:wknn-shapley-expression})} }
    \EndFor
\fi

\State

\State \cmt{Compute $\Gil$}

\State $\Gizero = \ind[w_i < 0]$.\footnote{Recall that we define $\U(S) = \ind \left[ \sum_{j=1}^{\min(K, |S|)} \widetilde{w}_{(j)} \ge 0 \right]$, hence $\U(\{z_i\}) - \U(\emptyset) \in \{-1, 0\}$ and is equal to $-1$ if and only if $w_i < 0$.}

\For{$\ell \in \{1, \ldots, K-1\}$}
\State
$\Gil = 
\begin{cases}
\sum_{m \in [N] \setminus i} \sum_{s \in [-\wtil_i, 0)} \Fi\left[m, \ell, s\right] & \text{for } y_i = \yval \\
\sum_{m \in [N] \setminus i} \sum_{s \in [0, -\wtil_i)} \Fi\left[m, \ell, s\right] & \text{for } y_i \ne \yval 
\end{cases}
$ \textcolor{blue}{ \texttt{(Theorem \ref{thm:Gil-and-Fi})} }
\EndFor

\State 

\State \cmt{Compute the Shapley value for $z_i$}

\State $\phi_{z_i} = \sign(w_i) \left[
\frac{1}{N} \sum_{\ell=0}^{K-1} \frac{\Gil}{{N-1 \choose \ell}} + 
\sum_{m = \max(i+1, K+1)}^{N} \frac{\Rim}{ m {m-1 \choose K} } 
\right]$.\footnote{
$\sign(w) = 
\begin{cases}
    1 & w>0 \\
    0 & w=0 \\
    -1 & w<0    
\end{cases}$.
} \textcolor{blue}{ \texttt{(Theorem \ref{thm:wknn-shapley-expression})} }

\EndFor
\end{algorithmic}
\end{algorithm}
\end{minipage}
}
}

%% file: pseudocode-approx.tex
\centerline{
\scalebox{0.8}{
\begin{minipage}[t]{1\linewidth}
\begin{algorithm}[H]
\ifinappendix
    \caption{Approximation of Weighted KNN-Shapley for binary classification}
    \label{alg:efficient-wknn-approx}
\else
    \caption{\red{Approximation of} Weighted KNN-Shapley for binary classification (reader-friendly version)}
    \label{alg:efficient-wknn-readable-approx}
\fi
\begin{algorithmic}[1]

\State \textbf{Input:} 
\begin{itemize}
    \item $K$ -- hyperparameter of weighted KNN algorithm. 
    \item $\zval = (\xval, \yval)$ -- the validation point. 
    \item $D = \{z_i = (x_i, y_i)\}_{i=1}^N$ -- sorted training set where $d(x_i, \xval) \le d(x_j, \xval)$ for any $i \le j$.
    \item \red{$\mstar$ -- hyperparameter for SV approximation (Section \ref{sec:deterministic-approx}). 
    $\mstar = N$ for exact SV calculation.} 
\end{itemize}

\State

\State Compute the weight $w_i = \weightfunc(x_i)$ for $i \in \{1, \ldots, N\}$. 

\State $\widetilde{w}_{j} = (2 \ind[\yval = y_{j}] - 1) w_{j}$ for $i \in \{1, \ldots, N\}$.

\State

\For{$i \in \{1, \ldots, N\}$}

\State

\State \cmt{Initialize $\Fi$}
\State Initialize $\Fi[m, \ell, s] = 0$ for $m \in \{1, \ldots, \red{\mstar} \}, \ell \in \{1, \ldots, K-1\}, s \in \wspace_{(K)}$. 

\For{$m \in \{ 1, \ldots, \red{\mstar} \} \setminus \{i\}$}
\State $\Fi[m, 1, \wtil_m] = 1$ \texttt{\textcolor{blue}{(Theorem \ref{thm:Fi-recursive-full})}}
\EndFor

\State 

\ifinappendix
    \State \cmt{Compute $\Fi$ \textcolor{blue}{(Runtime-optimized version)}}
    \For{$\ell \in \{2, \ldots, K-1\}$}
    \State $F_0[:] = \sum_{t=1}^{\ell-1} \Fi[t, \ell-1, :]$
    \For{$m \in \{ \ell, \ldots, \red{\mstar} \} \setminus \{i\}$}
    \For{$s \in \wspace_{(K)}$}
    \State $\Fi[m, \ell, s] = F_0[s - w_m]$
    \EndFor
    \EndFor
    \EndFor
\else
    \State \cmt{Compute $\Fi$ \textcolor{blue}{(Runtime-optimized version in Appendix \ref{appendix:pseudocode})}}
    
    \For{$\ell \in \{2, \ldots, K-1\}$}
    \For{$m \in \{ \ell, \ldots, \red{\mstar} \} \setminus \{i\}$}
    \For{$s \in \wspace_{(K)}$}
    \State $\Fi[m, \ell, s]
    = \sum_{t=1}^{m-1} \Fi[t, \ell-1, s - \wtil_m]$
    \EndFor
    \EndFor
    \EndFor
\fi

\State

\ifinappendix
    \State \cmt{Compute $\Rim$ \textcolor{blue}{(Runtime-optimized version)}}
    \For{$s \in \wspace_{(K)}$}
    \State 
    $R_0[s] = \sum_{t=1, t \ne i}^{ \max(i+1, K+1)-1 } \Fi[t, K-1, s]$. 
    \EndFor
    \For{$m \in \{ \max(i+1, K+1), \ldots, \red{\mstar} \}$}
    \State 
    $\Rim = 
    \begin{cases}
    \sum_{s \in [-\wtil_i, -\wtil_m)} R_0[s] & \text{for } y_i = \yval \\
    \sum_{s \in [-\wtil_m, -\wtil_i)} R_0[s] & \text{for } y_i \ne \yval 
    \end{cases}
    $
    \State $R_0 = R_0 + \Fi[m, K-1, :]$
    \EndFor
\else
    \State \cmt{Compute $\Rim$ \textcolor{blue}{(Runtime-optimized version in Appendix \ref{appendix:pseudocode})}}
    
    \For{$m \in \{ \max(i+1, K+1), \ldots, \red{\mstar} \}$}
    \State 
    $\Rim = 
    \begin{cases}
    \sum_{t=1}^{m-1} \sum_{s \in [-\wtil_i, -\wtil_m)} \Fi[t, K-1, s] & \text{for } y_i = \yval \\
    \sum_{t=1}^{m-1} \sum_{s \in [-\wtil_m, -\wtil_i)} \Fi[t, K-1, s] & \text{for } y_i \ne  \yval 
    \end{cases}
    $
    \EndFor
\fi

\State

\State \cmt{Compute $\Gil$}

\State $\widetilde{\texttt{G}}_{i,0}^{(\mstar)} = \ind[w_i < 0]$.\footnote{Recall that we define $\U(S) = \ind \left[ \sum_{j=1}^{\min(K, |S|)} \widetilde{w}_{(j)} \ge 0 \right]$, hence $\U(\{z_i\}) - \U(\emptyset) \in \{-1, 0\}$ and is equal to $-1$ if and only if $w_i < 0$.}

\For{$\ell \in \{1, \ldots, K-1\}$}
\State
$\Giltil^{(\mstar)} = 
\begin{cases}
\sum_{m \in [\red{\mstar}] \setminus i} \sum_{s \in [-\wtil_i, 0)} \Fi\left[m, \ell, s\right] & \text{for } y_i = \yval \\
\sum_{m \in [\red{\mstar}] \setminus i} \sum_{s \in [0, -\wtil_i)} \Fi\left[m, \ell, s\right] & \text{for } y_i \ne \yval 
\end{cases}
$ \texttt{\textcolor{blue}{(Definition \ref{def:approxmstar})}}
\EndFor

\State 

\State \cmt{Compute the Shapley value for $z_i$}

\State $\phihatmstar = \sign(w_i) \left[
\frac{1}{N} \sum_{\ell=0}^{K-1} \frac{\Giltil^{(\mstar)}}{{N-1 \choose \ell}} + 
\sum_{m = \max(i+1, K+1)}^{\red{\mstar}} \frac{\Rim}{ m {m-1 \choose K} } 
\right]$.\footnote{
$\sign(w) = 
\begin{cases}
    1 & w>0 \\
    0 & w=0 \\
    -1 & w<0    
\end{cases}$. 
} \texttt{\textcolor{blue}{(Definition \ref{def:approxmstar})}}

\EndFor
\end{algorithmic}
\end{algorithm}
\end{minipage}
}
}

%% file: appendix_multiclass.tex
In Section \ref{sec:shapley-for-binary}, we have developed efficient solutions for computing and approximating the Data Shapley of weighted hard-label KNN \emph{binary} classifiers. Directly adapting the same techniques from the binary to multi-class classification setting, however, can significantly increase the overall time complexity (see Appendix \ref{sec:multi-class-exact} for details). Hence, in Appendix \ref{sec:multiclass-heuristic}, we introduce an alternative utility function to measure the performance of WKNN classifiers. This new utility function provides more detailed confidence information about the KNN classifier, as opposed to the original utility function in (\ref{eq:util}), which only provides basic zero-one correctness. More importantly, computing the Data Shapley in terms of the new utility function can be conveniently reduced to the WKNN-Shapley computation for binary classifiers without significantly increasing time complexity, a benefit leveraged from the linearity axiom of the Shapley value. 

\textbf{Notation Review.} 
Recall that we use $\wspace$ to denote the discretized space of $[0, 1]$, where we create $2^b$ equally spaced points within the interval when we use $b$ bits for discretization. We denote $\nwspace := |\wspace| = 2^b$ the size of the weight space. Furthermore, we use $\wspace_{(K)}$ to denote the discretized space of $[0, K]$ (where we create $K 2^b$ equally spaced points within the interval). 
We use $\NBK(S)$ to denote the set of data points that is within the $K$-nearest neighbors of $\xval$ among $S$. 
We use $\alpha_{\xval}^{(S, j)}$ denotes the index (among $D$) of $j$th closest data point in $S$ to $\xval$.




\subsection{Direct Extension from Binary Classification Setting can be Inefficient}
\label{sec:multi-class-exact}


We first discuss a simple, direct extension of our exact WKNN-Shapley algorithm from binary to multi-class classification setting. 
In Algorithm \ref{alg:efficient-wknn-readable}, the main idea is to maintain a record of $\Fi[m, \ell, s]$ for a singular scalar value $s$ which represents the summation of ``signed weights'' $\wtil_j$. In order to extend this approach to the multi-class setting, it is natural to enhance this scalar representation to a ``histogram'' depiction, $\Fi[m, \ell, \s]$, where $\s$ is the vector sum of weights for each data point, and the weights are in the form of one-hot encoding. 
That is, in the multi-class setting, $\Fi$ is augmented to record the number of subsets such that the sum of weights of the data points in the one-hot encoding is equal to the histogram $\s$ (subject to the conditions analog to those in Definition \ref{def:Fi-binary}). 
Denote $\e_y \in [C]$ the one-hot encoding of the label, with 1 on $y$th entry and 0 otherwise. 
The augmented $\Fi$ is defined as follows:

\begin{definition}
\label{def:Fi-multiclass}
Let $\Fi\left[m, \ell, \s \right]$ denote the count of subsets $S \subseteq D \setminus \{z_i\}$ of size $\ell$ that satisfy the conditions below: \textbf{(1)} $\condKNN$, \textbf{(2)} Within $S$, the data point $x_m$ is the $\min(\ell, K)$-th closest to the query example $\xval$, and \textbf{(3)} $\sum_{j=1}^{\min(\ell, K-1)} w_{\alpha_{\xval}^{(S, j)}} \e_{y_{ \alpha_{\xval}^{(S, j)}}} = \s$. 
\end{definition}

The results in the maintext (Theorem \ref{thm:Gil-and-Fi}, \ref{thm:Fi-recursive}, and \ref{thm:wknn-shapley-expression}) can be easily adapted to this more generalized definition of $\Fi$. 
While this direct extension can compute the exact Data Shapley for the utility function in (\ref{eq:util}), it has a time complexity of $O(K^{1+C} N^2 \nwspace^C)$ as we need to record $\Fi[m, \ell, \s]$ for all possible histograms $\s \in \wspace_{(K)}^C$, where $\wspace_{(K)}^C$ denotes the product space of $\wspace_{(K)}$. This is manageable for datasets with a modest size of class space. However, for datasets with a large class space, this complexity can render the runtime prohibitively large. 


\subsection{Alternative Utility Function that Enables More Efficient Computation of WKNN-Shapley}
\label{sec:multiclass-heuristic}

Due to the above-mentioned computational bottleneck, we introduce an alternative utility function for weighted KNN classifiers, which not only reflects the KNN classifiers' performance but also paves the way for a more efficient Data Shapley computation analogous to that of the binary setting. 
This is another instantiation of the rationale of developing advanced data valuation techniques with proper adjustment to the utility function. 


\textbf{Alternative Utility Function for Weighted Hard-Label KNN Classifiers.} 
For a class $c \in [C] \setminus \{\yval\}$, we denote 
\begin{align}
\U^{(c)}(S; \zval) 
:=
\ind\left[ 
\sum_{j=1}^{\min(K, |S^{(c)}|)} w_{\alpha_{\xval}^{(S^{(c)}, j)}} \ind \left[ y_{\alpha_{\xval}^{(S^{(c)}, j)}} = \yval \right] \ge \sum_{j=1}^{\min(K, |S^{(c)}|)} w_{\alpha_{\xval}^{(S^{(c)}, j)}} \ind \left[y_{\alpha_{\xval}^{(S^{(c)}, j)}} = c \right]
\right]
\end{align}
where 
$
S^{(c)} 
:= \{(x, y) \in S: y \in \{\yval, c\} \}
$ is the subset of $S$ whose labels are either $\yval$ and $c$. 
We propose an alternative utility function as follows:
\begin{align}
\Ualter(S; \zval) 
:= 
\frac{1}{C-1}
\sum_{c \in [C] \setminus \yval} 
\U^{(c)}(S^{(c)}; \zval)
\label{eq:util-multiclass}
\end{align}
Note that for binary classifiers, the new utility function $\Ualter$ reduces to the original $\U$. 
\textbf{Interpretation of the Alternative Utility Function:} 
The alternative utility function, $\Ualter$, captures a more fine-grained view of the classifier's performance. Instead of just deciding based on whether a prediction is correct as the original utility function in (\ref{eq:util}), it reduces the multi-class classification game as multiple binary-class classification games, and assesses how many times the correct class, $\yval$, are correctly being predicted in those subgames. 
Hence, $\Ualter$ provides insight into not just the correctness, but also the relative confidence of a prediction with respect to other classes. 

\textbf{Data Shapley for $\Ualter$.}
The linearity axiom of the Shapley value provides that 
\begin{align*}
\phi_{z_i}( \Ualter ) = 
\frac{1}{C-1} \sum_{c \in [C] \setminus \yval} \phi_{z_i}(\U^{(c)})
\end{align*}
Denote the subset $D_{\yval, c} \subseteq D$ such that $D_{\yval, c}:= \{(x, y) \in D: y \in \{\yval, c\} \}$. 
Observe that $\phi_{z_i}(\U^{(c)}) = 0$ for all $z_i \notin D_{\yval, c}$, and $\phi_{z_i}(\U^{(c)})$ for $z_i \in D_{\yval, c}$ can be easily computed with WKNN-Shapley for binary classification setting (Algorithm \ref{alg:efficient-wknn}). Hence, we can first compute $\phi_{z_i}(\U^{(c)})$ for each $c \in [C] \setminus \yval$ individually, and then aggregate these values. 

\input{pseudocode-multiclass-new}

\begin{remark}
We can also replace the exact computation of $\phi_{z_i}(\U^{(c)})$ in line 4 by our deterministic approximation algorithm (Algorithm \ref{alg:efficient-wknn-approx}) and obtain a deterministic approximation for $\phi_{z_i}( \Ualter )$. 
The error bound is simply the average of error bounds for computing $\phi_{z_i}(\U^{(c)})$ for each of $c \in [C] \setminus \{\yval\}$. 
\end{remark}

While this might imply an inevitable factor of $C$ in the computational complexity, note that the ``effective dataset'' for $\U^{(c)}$ is the subset $D_{\yval, c} \subseteq D$ that comprises only data points labeled $\yval$ or $c$. 
As a result, the computational time to compute the Shapley value for $\U^{(c)}$ reduces to $O(K^2 |D_{\yval, c}|^2 \nwspace)$. This provides a huge runtime saving when the dataset is balanced. 

\begin{theorem}
For a class-balanced training dataset $D$ with $C$ classes, computing the exact WKNN-Shapley $\{\phi_{z_i}(\Ualter)\}_{z_i \in D}$ achieves runtime $O\left(\frac{K^2 N^2 \nwspace}{C}\right)$. 
\label{thm:runtime-for-multiclass-heuristic}
\end{theorem}
\begin{proof}
When the dataset $D$ is balanced, we have $|D_{\yval, c}| = \frac{2N}{C}$. 
Hence, the runtime of computing $\{\phi_{z_i}(\U^{(c)})\}_{z_i \in D}$ is $O\left(\frac{K^2 N^2 \nwspace}{C^2}\right)$, and hence the total runtime is $O\left(\frac{K^2 N^2 \nwspace}{C}\right)$. 
\end{proof}

Remarkably, this methodology is even more efficient than its binary classification counterpart.

%% file: pseudocode-multiclass-new.tex
\begin{algorithm}[H]
\caption{Weighted KNN-Shapley for multi-class classification}
\begin{algorithmic}[1]

\State \textbf{Input:} 
\begin{itemize}
    \item $K$ -- hyperparameter of weighted KNN algorithm. 
    \item $\zval = (\xval, \yval)$ -- the validation point. 
    \item $D = \{z_i = (x_i, y_i)\}_{i=1}^N$ -- sorted training set where $d(x_i, \xval) \le d(x_j, \xval)$ for any $i \le j$.
    \item $C$ -- number of classes
\end{itemize}

\For{$c \in [C] \setminus \{ \yval \}$}
\State $\phi_{z_i}(\U^{(c)}) = 0$ for $z_i \notin D_{\yval, c}$.  
\State Compute $\phi_{z_i}(\U^{(c)})$ for $z_i \in D_{\yval, c}$ by executing Algorithm \ref{alg:efficient-wknn} on $D_{\yval, c}$. 
\EndFor

\State \textbf{Return:} $\phi_{z_i}( \Ualter ) = 
\frac{1}{C-1} \sum_{c \in [C] \setminus \yval} \phi_{z_i}(\U^{(c)})$ for $z_i \in D$. 

\end{algorithmic}
\end{algorithm}

%% file: appendix_proofs.tex
\textbf{Notation Review.} 
Recall that we use $\wspace$ to denote the discretized space of $[0, 1]$, where we create $2^b$ equally spaced points within the interval when we use $b$ bits for discretization. We denote $\nwspace := |\wspace| = 2^b$ the size of the weight space. Furthermore, we use $\wspace_{(K)}$ to denote the discretized space of $[0, K]$ (where we create $K 2^b$ equally spaced points within the interval). We use $\NBK(S)$ to denote the set of data points that is within the $K$-nearest neighbors of $\xval$ among $S$. We use $\alpha_{\xval}^{(S, j)}$ denotes the index (among $D$) of $j$th closest data point in $S$ to $\xval$. 

\begin{theorem}[Restate of Theorem \ref{thm:marginal-contri}]
For any data point $z_i \in D$ and any subset $S \subseteq D \setminus \{z_i\}$, the marginal contribution has the expression as follows:
\begin{align}
\U(S \cup \{z_i\}) - \U(S) =
\begin{cases} 
1 & \text{if~}y_i = \yval, \condKNN, \condNegToPos \\
-1 & \text{if~}y_i \ne \yval, \condKNN, \condPosToNeg \\
0 & \text{Otherwise}
\end{cases}
\end{align}
where 
\begin{align}
\condKNN &:= z_i~\text{is within}~K~\text{nearest neighbors of}~\xval~\text{among}~S \cup \{z_i\} \\
\condNegToPos &:= 
\begin{cases}
    \sum_{z_j \in S} \wtil_j \in [-\wtil_i, 0) & \text{if}~|S| \le K-1 \\
    \sum_{j=1}^{K-1} \wtil_{\alpha_{\xval}^{(S, j)}} \in \left[ -w_i,  -\wtil_{\alpha_{\xval}^{(S, K)}}\right) & \text{if}~|S| \ge K
\end{cases} \\
\condPosToNeg &:= 
\begin{cases}
    \sum_{z_j \in S} \wtil_j \in [0, -\wtil_i) & \text{if}~|S| \le K-1 \\
    \sum_{j=1}^{K-1} \wtil_{\alpha_{\xval}^{(S, j)}} \in \left[-\wtil_{\alpha_{\xval}^{(S, K)}}, -w_i\right) & \text{if}~|S| \ge K
\end{cases}
\end{align}
\end{theorem}
\begin{proof}
First of all, we observe that if $z_i \notin \NBK(S \cup \{z_i\})$, i.e., if $z_i$ is not within the $K$ nearest neighbors of the queried example $\xval$ among the subset $S \cup \{z_i\}$, then the prediction of KNN classifier does not change, and hence we know that $\U(S \cup \{z_i\}) = \U(S)$. Hence $\condKNN$ is necessary for $\U(S \cup \{z_i\}) - \U(S)$ to be non-zero. 

If $z_i \in \NBK(S \cup \{z_i\})$, we divide into two cases: 
\circled{1} If $|S| \le K-1$ we know that adding $z_i$ will not exclude any other data point from the $K$ nearest neighbors of $\xval$. 
Hence $\U(S \cup \{z_i\}) - \U(S) = 1$ only if $y_i = \yval$ and $\sum_{z_j \in S} \wtil_j \in [-\wtil_i, 0)$, 
and $\U(S \cup \{z_i\}) - \U(S) = -1$ only if $y_i \ne \yval$ and $\sum_{z_j \in S} \wtil_j \in [0, -\wtil_i)$. 
\circled{2} If $|S| \ge K$ we know that adding $z_i$ will exclude the original $K$th nearest neighbors of $\xval$ among dataset $S$. 
Hence, $\U(S \cup \{z_i\}) - \U(S) = 1$ only if $y_i = \yval$ and $\sum_{j=1}^{K-1} \wtil_{\alpha_{\xval}(S, j)} \in \left[ -w_i,  -\wtil_{\alpha_{\xval}(S, K)}\right)$, and 
$\U(S \cup \{z_i\}) - \U(S) = -1$ only if $y_i \ne \yval$ and $\sum_{j=1}^{K-1} \wtil_{\alpha_{\xval}(S, j)} \in \left[-\wtil_{\alpha_{\xval}(S, K)}, -w_i\right)$. 
\end{proof}


\begin{theorem}[Full version of Theorem \ref{thm:Fi-recursive}]
\label{thm:Fi-recursive-full-withproof}
When $K > 1$,\footnote{Since \cite{jia2019efficient} has shown that weighted KNN-Shapley can be computed in $O(N^K)$ time complexity, we focus on the setting where $K > 1$.} for $\ell = 1$ we have
$
\Fi[m, 1, s] = 
\begin{cases} 
1 & s = w_m \\
0 & s \ne w_m
\end{cases}
$. 
We can then compute $\Fi[m, \ell, s]$ for $\ell \ge 2$ with the following relations: 

If $\ell \le K-1$, we have
\begin{align}
    \Fi[m, \ell, s] = \sum_{t=1}^{m-1} \Fi[t, \ell-1, s - w_m]
\label{eq:recursive-proof}
\end{align}
and if $\ell \ge K$, we have 
\begin{align}
    \Fi[m, \ell, s] = 
    \begin{cases}
        0 & m < i \\
        \sum_{t=1}^{m-1} \Fi[t, K-1, s] {N-m \choose \ell-K} & m > i
    \end{cases}
\label{eq:explicit-proof}
\end{align}
Note that we set $\Fi[i, \cdot, \cdot] = 0$ for mathematical convenience. 
\end{theorem}
\begin{proof}

\textbf{Base case of $\ell=1$:} by definition of $\Fi$, if $\ell=1$ the dataset that satisfy the conditions required for $\Fi[m, 1, \cdot]$ can only be the singleton $\{z_m\}$, and hence the base case is straightforward. 

\textbf{Case of $\ell \le K-1$:} the inclusion of $z_i$ in $\NBK(S \cup \{z_i\})$ is guaranteed for this range of $\ell$. The datasets that satisfy the conditions required for $\Fi[m, \ell, s]$ can be partitioned based on the $(\ell-1)$th closest point to $\xval$, which leads to the recursive relation in (\ref{eq:recursive-proof}). 

\textbf{Case of $\ell \ge K$:} 
Since $x_m$ is the $K$-th nearest data point to $\xval$ within $S$, we have $\Fi[m, \ell, s] = 0$ for any $m < i$ since $z_i$ is also required to be within $K$ nearest neighbor to $\xval$. 
When $m > i$, the datasets that satisfy the conditions required for $\Fi[m, \ell, s]$ can be partitioned based on the $(K-1)$th closest point to $\xval$, which leads to the relation in (\ref{eq:explicit-proof}) by simple combinatorial analysis. 
\end{proof}


\begin{theorem}[Restate of Theorem \ref{thm:wknn-shapley-expression}]
For a weighted, hard-label KNN binary classifier using the utility function given by (\ref{eq:util-binary}), the Shapley value of data point $z_i$ can be expressed as:
\begin{align}
    \phi_{z_i} 
    = 
    \sign(w_i) \left[
    \frac{1}{N} \sum_{\ell=0}^{K-1} \frac{\Gil}{{N-1 \choose \ell}} + \sum_{m = \max(i+1, K+1)}^N \frac{\Rim}{m {m-1 \choose K}} \right]
\end{align}
where 
\begin{align}
\Rim := 
\begin{cases}
    \sum_{t=1}^{m-1} \sum_{s \in [-\wtil_i, -\wtil_m)} \Fi[t, K-1, s] & \text{for } y_i = \yval \\
    \sum_{t=1}^{m-1} \sum_{s \in [-\wtil_m, -\wtil_i)} \Fi[t, K-1, s] & \text{for } y_i \ne  \yval 
\end{cases}
\end{align}

\end{theorem}
\begin{proof}

We state the proof for the case where $y_i = \yval$, and the proof for the case where $y_i \ne \yval$ is nearly identical. Recall that
\begin{align*}
\Gil = 
\begin{cases} 
\sum_{m \in [N] \setminus i} \sum_{s \in [-\wtil_i, 0)} \Fi\left[m, \ell, s\right] & 
\ell \le K-1 \\
\sum_{m \in [N] \setminus i} \sum_{s \in [-\wtil_i, -\wtil_m)} \Fi\left[m, \ell, s\right] & \ell \ge K
\end{cases}
\end{align*}
if $y_i = \yval$. 

When $\ell \ge K$, we have 
\begin{align*}
\Gil 
&= \sum_{m \in [N] \setminus i} \sum_{s \in [-\wtil_i, -\wtil_m)} \Fi\left[m, \ell, s\right] \\
&= \sum_{m = \max(i+1, K+1)}^N \sum_{s \in [-\wtil_i, -\wtil_m)} \Fi\left[m, \ell, s\right] \\
&= \sum_{m = \max(i+1, K+1)}^N \sum_{s \in [-\wtil_i, -\wtil_m)} {N-m \choose \ell-K}
\sum_{t=1, t \ne i}^{m-1} \Fi[t, K-1, s] \\
&= \sum_{m = \max(i+1, K+1)}^N {N-m \choose \ell-K} 
\sum_{s \in [-\wtil_i, -\wtil_m)}
\sum_{t=1, t \ne i}^{m-1} \Fi[t, K-1, s] \\
&= \sum_{m = \max(i+1, K+1)}^N {N-m \choose \ell-K} \Rim 
\end{align*}
where $\Rim = \sum_{s \in [-\wtil_i, -\wtil_m)}
\sum_{t=1, t \ne i}^{m-1} \Fi[t, K-1, s]$. 

\begin{align*}
\sum_{\ell = K}^{N-1} \frac{\Gil}{{N-1 \choose \ell}}
&= 
\sum_{\ell = K}^{N-1} 
\sum_{m = \max(i+1, K+1)}^N 
\frac{{N-m \choose \ell-K} \Rim}{ {N-1 \choose \ell} } \\
&= 
\sum_{m = \max(i+1, K+1)}^N
\Rim 
\sum_{\ell = K}^{N-1} 
\frac{{N-m \choose \ell-K}}{ {N-1 \choose \ell} } \\
&= 
\sum_{m = \max(i+1, K+1)}^N
\Rim 
\left(
\sum_{\ell = K}^{N-1} 
\frac{ {m-1 \choose K} {N-m \choose \ell-K}}{ {N-1 \choose \ell} } \right)
{m-1 \choose K}^{-1} \\
&= 
\sum_{m = \max(i+1, K+1)}^N
\Rim \left( \frac{N}{m} \right) {m-1 \choose K}^{-1}
\end{align*}
\end{proof}

\begin{theorem}[Restate of Theorem \ref{thm:runtime-exact}]
Algorithm \ref{alg:efficient-wknn} (in Appendix \ref{appendix:pseudocode}) computes the exact Shapley value and achieves $O(K^2 N^2 \nwspace)$ time complexity.
\end{theorem}
\begin{proof}
It is easy to see that the for-loop for computing $\Fi$ for $\ell \le K$ requires a runtime of $O(KN |\wspace_{(K)}|)$. The for-loop for computing $\Rim$ requires a runtime of $O(N |\wspace_{(K)}|)$. The for-loop for computing $\Gil$ for $\ell \le K$ requires a runtime of $O(KN |\wspace_{(K)}|)$. 
All of these subroutines are included in the outside for-loop for computing $\phi_{z_i}$ for all $z_i \in D$. 
Hence, the overall runtime is $O(KN^2 |\wspace_{(K)}|) = O(K^2 N^2 \nwspace)$. 
\end{proof}

\begin{theorem}[Restate of Theorem \ref{thm:runtime-approx}]
Algorithm \ref{alg:efficient-wknn-approx} (in Appendix \ref{appendix:pseudocode-approx}) computes the approximated WKNN-Shapley $\widehat \phi_{z_i}^{(\mstar)}$ for all $z_i \in D$ and achieves a total runtime of $O(\nwspace K^2 N \mstar)$.
\end{theorem}
\begin{proof}
It is easy to see that the for-loop for computing $\Fi$ for $\ell \le K$ requires a runtime of $O(K \mstar |\wspace_{(K)}|)$. The for-loop for computing $\Rim$ requires a runtime of $O(\mstar |\wspace_{(K)}|)$. The for-loop for computing $\Gil$ for $\ell \le K$ requires a runtime of $O(K\mstar |\wspace_{(K)}|)$. 
All of these subroutines are included in the outside for-loop for computing $\phi_{z_i}$ for all $z_i \in D$. 
Hence, the overall runtime is $O(KN \nwspace |\wspace_{(K)}|) = O(\nwspace K^2 N \mstar)$. 
\end{proof}

\begin{theorem}[Restate of Theorem \ref{thm:error-bound}]
For any $z_i \in D$, the approximated Shapley value $\phihatmstar$ \textbf{(1)} shares the same sign as $\phi_{z_i}$, \textbf{(2)} ensures $\left| \phihatmstar \right| \le \left| \phi_{z_i} \right|$, and \textbf{(3)} has the approximation error bounded by $\left| \phihatmstar - \phi_{z_i} \right| \le \eps(\mstar)$ where 
\begin{align*}
\eps(\mstar) := \sum_{m = \mstar+1}^N \left( \frac{1}{m-K} - \frac{1}{m} \right) + \sum_{\ell=1}^{K-1} \frac{ {N \choose \ell} - {\mstar \choose \ell} }{ N {N-1 \choose \ell} } = O\left( K / \mstar \right)
\end{align*}
\end{theorem}
\begin{proof}
The property \textbf{(1)} follows from that the sign of both exact and approximated WKNN-Shapley only depends on whether $y_i = \yval$. 
The property \textbf{(2)} follows from the approximation algorithm only counts \emph{part} of the subproblems, and hence $\left| \phihatmstar \right| \le \left| \phi_{z_i} \right|$. We now prove property \textbf{(3)}. 

In the exact algorithm \ref{alg:efficient-wknn-readable}, we have
\begin{align*}
\phi_i = 
\underbrace{\frac{1}{N} \sum_{\ell=0}^{K-1} \frac{\Gil}{{N-1 \choose \ell}}}_{(A)}
+ 
\underbrace{
\sum_{m = \max(i+1, K+1)}^N \Rim \left(\frac{1}{m}\right) {m-1 \choose K}^{-1} }_{(B)}
\end{align*}
First of all, note that 
\begin{align*}
\sum_{s \in \wspace} \Fi\left[m, \ell, s\right]
= 
{m-1 - \ind[i < m] \choose \ell-1}
\le 
{m-1 \choose \ell-1}
\end{align*}
for any $\ell \le K$ since $\sum_{s \in \wspace} \Fi\left[m, \ell, s\right]$ is essentially the total number of subsets $S \subseteq D \setminus {z_i}$ of size $\ell$ where $z_m$ is the farthest data point to the query example $\xval$. 

Now, denote 
\begin{align*}
\Giltil := \sum_{m=1}^{\mstar} \sum_{s \in [-\wtil_i, 0)} \Fi\left[m, \ell, s\right]
\end{align*}
for $1 \le \ell \le K-1$. The gap between $\Gil$ and $\Giltil$ can be bounded as follows:
\begin{align*}
\left| \Giltil - \Gil \right|
&= 
\sum_{m = \mstar+1}^N \sum_{s \in [-\wtil_i, 0)} \Fi\left[m, \ell, s\right] \\
&\le 
\sum_{m = \mstar+1}^N \sum_{s \in \wspace} \Fi\left[m, \ell, s\right] \\
&\le \sum_{m = \mstar+1}^N {m-1 \choose \ell-1} \\
&= \sum_{m = \ell}^N {m-1 \choose \ell-1} - \sum_{m = \ell}^{\mstar} {m-1 \choose \ell-1} \\
&= {N \choose \ell} - {\mstar \choose \ell}
\end{align*}

Now we bound the error from taking the approximation $\Rimhat = 0$ for $m \ge \mstar+1$. Since we have 
\begin{align*}
\Rim 
&= 
\sum_{t=1}^{m-1} \sum_{s \in [-\wtil_i, -\wtil_m)} \Fi[t, K-1, s] \\
&\le \sum_{t=1}^{m-1} {t-1 \choose K-2} \\
&= {m-1 \choose K-1}
\end{align*}
Hence 
\begin{align*}
\sum_{m = \max(i+1, K+1, \mstar+1)}^N \Rim \left(\frac{1}{m}\right) {m-1 \choose K}^{-1} 
&\le 
\sum_{m = \max(i+1, K+1, \mstar+1)}^N 
{m-1 \choose K-1} \left(\frac{1}{m}\right) {m-1 \choose K}^{-1} \\
&\le 
\sum_{m = \mstar+1}^N 
{m-1 \choose K-1} \left(\frac{1}{m}\right) {m-1 \choose K}^{-1} \\
&= \sum_{m = \mstar+1}^N \frac{K}{m (m-K)} \\
&= \sum_{m = \mstar+1}^N \left( \frac{1}{m-K} - \frac{1}{m} \right)
\end{align*}

Hence, for any data point $z_i$, we have 
\begin{align*}
\left| \phihatmstar - \phi_i \right| 
&= 
\frac{1}{N} \sum_{\ell=0}^{K-1} \frac{\left| \Gil - \Giltil \right|}{{N-1 \choose \ell}}
+ 
\sum_{m = \max(i+1, K+1, \mstar+1)}^N \Rim \left(\frac{1}{m}\right) {m-1 \choose K}^{-1} \\
&\le 
\frac{1}{N} \sum_{\ell=1}^{K-1} \frac{{N \choose \ell} - {\mstar \choose \ell}}{ {N-1 \choose \ell} } 
+
\sum_{m = \mstar+1}^N \left( \frac{1}{m-K} - \frac{1}{m} \right)
\end{align*}
\end{proof}



\begin{theorem}[Restate of Theorem \ref{thm:approx-shapley-property}]
The approximated Shapley value $\{ \phihatmstar \}_{z_i \in D}$ satisfies symmetry and null player axiom. 
\end{theorem}
\begin{proof}

\textbf{Null Player Axiom.} If a data point $z_i$ is a null player (i.e., $\U(S \cup z_i) = \U(S)$ for all $S \subseteq D \setminus \{z_i\}$), then we have $\phi_{z_i} = 0$. 
From the expression in Theorem \ref{thm:wknn-shapley-expression}, we must have $\Rim = 0$ for all $0 \le m \le N$ and $\Gil = 0$ for all $0 \le \ell \le N-1$ (as these are non-negative quantities). 
Since $\Giltil \le \Gil$, we know that $\Giltil = 0$ for all $0 \le \ell \le N-1$. Hence, we have $\phihatmstar = 0$.

\newcommand{\condM}[1]{\texttt{cond}_{\le #1}}

\textbf{Symmetry Axiom.} 
Denote the condition $\condM{m}$ as ``within $S$, the $\min(\ell, K)$-th closest to the query example $\xval$ is among $\{x_k\}_{k=1}^{m}$. 
If two data points $z_i, z_j$ are symmetry (i.e., $\U(S \cup z_i) = \U(S \cup z_j)$ for all $S \subseteq D \setminus \{z_i, z_j\}$), we must have 
$\Giltilmstar = \widetilde{\texttt{G}}_{j, \ell}^{(\mstar)}$ since 
\begin{align*}
\Giltilmstar = \sum_{S \subseteq D \setminus \{z_i\}, |S|=\ell, \condM{\mstar} } [\U(S \cup \{z_i\}) - \U(S)]
\end{align*}
Furthermore, we also have $\Rim = \texttt{R}_{j, m}$ since 
\begin{align*}
\Rim = \sum_{S \subseteq D \setminus \{z_i\}, |S|=K-1, \condM{m-1} } [\U(S \cup \{z_i\}) - \U(S)]
\end{align*}
which leads to 
$
\phihat_{z_i}^{(\mstar)} = \phihat_{z_j}^{(\mstar)}
$. 
\end{proof}

%% file: appendix_evalsetting.tex
We provide a summary of the content in this section for the convenience of the readers. 
\begin{itemize}
    \item Appendix \ref{appendix:eval-settings}: General experiment settings (datasets and implementation details). 
    \item Appendix \ref{sec:eval-discretization}: Additional experiments for evaluating the influence of weights discretization. 
    \item Appendix \ref{appendix:runtime}: Experiment Settings and additional results for the runtime comparison of algorithms for computing/approximating WKNN-Shapley. 
    \item Appendix \ref{appendix:eval-noisydetect}: Experiment Settings and additional results for noisy data detection task. 
    \item Appendix \ref{appendix:eval-varyEPS}: Ablation study for the choice of $\mstar$ for the deterministic approximation algorithm from Section \ref{sec:deterministic-approx}. 
    \item Appendix \ref{appendix:eval-comparison}: Qualitative comparison between weighted and unweighted KNN-Shapley scores. 
    \item Appendix \ref{appendix:eval-dataselection}: Additional results for data selection tasks. 
\end{itemize}

\newpage

\subsection{Experiment Settings}
\label{appendix:eval-settings}

\subsubsection{Datasets}
\label{appendix:datasets}

An overview of the dataset information can be found in Table \ref{tb:datasets}. 
Following the existing literature in data valuation \cite{ghorbani2019data, kwon2022beta, jia2019towards, wang2023data, wang2023threshold}, we preprocess datasets for ease of training. 
Following \cite{kwon2022beta}, for Fraud, Creditcard, and all datasets from OpenML, we subsample the dataset to balance positive and negative labels. For these datasets, if they have multi-class, we binarize the label by considering $\ind[y=1]$. 
Following \cite{wang2023threshold}, for the image dataset MNIST, CIFAR10, we apply a ResNet50 \cite{he2016deep} that is pre-trained on the ImageNet dataset as the feature extractor. This feature extractor produces a 1024-dimensional vector for each image. 
For the sentence classification datasets AGNews and DBPedia, we use sentence embedding \cite{reimers-2019-sentence-bert} to extract features, resulting in 1024-dimensional vectors for each textual sample. We then standardize these extracted features using L2 normalization.

The size of each dataset we use is shown in Table \ref{tb:datasets}. For some of the datasets, we use a subset of the full set. The validation data size we use is 10\% of the training data size. 


\begin{table}[h]
\centering
\begin{tabular}{@{}cccc@{}}
\toprule
\textbf{Dataset} & \textbf{Number of classes} & \textbf{Size of dataset} &  \textbf{Source}                        \\ \midrule
Click            & 2 &2000 &\url{https://www.openml.org/d/1218}  \\
Fraud            & 2&2000 &\cite{dal2015calibrating}            \\
Creditcard       & 2&2000 &\cite{yeh2009comparisons}            \\
Apsfail          & 2&2000 &\url{https://www.openml.org/d/41138} \\
Phoneme          & 2&2000 &\url{https://www.openml.org/d/1489}  \\
Wind             & 2&2000 &\url{https://www.openml.org/d/847}   \\
Pol              & 2&2000 &\url{https://www.openml.org/d/722}   \\
CPU              & 2&2000 &\url{https://www.openml.org/d/761}   \\
2DPlanes         & 2&2000 &\url{https://www.openml.org/d/727}    \\ 
MNIST            &10 &2000 &\cite{lecun1998mnist}                \\
CIFAR10          &10 &2000 &\cite{krizhevsky2009learning}        \\
AGnews           &4  &2000 &\cite{wang2021n24news}                \\
DBPedia          &14 &2000 &\cite{auer2007dbpedia}        \\ \bottomrule
\end{tabular}
\caption{A summary of datasets used in Section \ref{sec:eval}'s experiments.}
\label{tb:datasets}
\end{table}

\subsubsection{Implementation of Weighted KNN-Shapley}
In Section \ref{sec:eval} in the main text, the weights used in KNN are based on $\ell_2$ distance between the training point and queried example, and then normalize all weights to $[0, 1]$. That is, the weight of a data point $z_i$ is computed by  
\begin{align*}
w_i := \frac{\norm{x_N - \xval}-\norm{x_i - \xval}}{ \norm{x_N - \xval} - \norm{x_1 - \xval} }
\end{align*} 
The weights are then discretized by rounding to the nearest values that can be represented with $b$ bits. 
That is, we create $2^b$ equally spaced points within the interval of $[0, 1]$, and round the weights to the closest point in the discretized space. We set the number of bits $b = 3$ in all experiments unless explicitly specified. 

\subsubsection{Details for Mislabel Data Detection Experiment}

In the experiment of mislabeled data detection, we randomly choose 10\% of the data points and flip their labels. Specifically, we flip 10\% of the labels by picking an alternative label from the rest of the classes uniformly at random. 
\begin{remark}[\textbf{Baseline of Data Shapley}]
We note that several works have demonstrated that (unweighted) KNN-Shapley significantly outperforms the traditional Data Shapley \cite{pandl2021trustworthy, wang2023threshold} in discerning data quality in tasks such as mislabeled data detection. Moreover, Data Shapley is highly inefficient as it requires ML models for many times, which is impractical for actual use. Hence, in this work, we omit the baseline of Data Shapley.  
\end{remark}




%% file: appendix_evalexp.tex
\subsection{Error From Weight Discretization}
\label{sec:eval-discretization}

\subsubsection{Value Deviation}

We empirically study the difference between WKNN-Shapley computed based on the original continuous weights and the discretized weights. However, for continuous weights, it is computationally infeasible to compute the exact Data Shapley. Therefore, we instead look at the computed Shapley values' difference when using $b$ bits and $b+1$ bits for $b = 1, 2, \ldots$. 
Figure \ref{fig:discretization_error} shows the results for $\ell_2$ and $\ell_\infty$ error. 
That is, a point on the figure at x-axis $b$ refers to the \emph{error reduction} if using one more bit $b+1$.
We have two observations here: 
\textbf{(1)} The error converges quickly as $b$ increases and is near zero after $b \ge 5$. 
\textbf{(2)} The larger the dataset size $N$ is, the smaller the error is. 
This interesting phenomenon is because the errors are dominated by the differences in the Shapley value computed for influential data points. 
When the dataset size is small, there are more influential data points since the performance of models trained on different data subsets can be significantly different from each other. 
On the other hand, when the dataset size is larger, there will be fewer influential points since most of the data subsets have a high utility (see Figure \ref{fig:shapley-distribution} for the visualization of the comparison between the distribution of data value scores). 

\begin{figure}[h]
\centering
\includegraphics[width=\columnwidth]{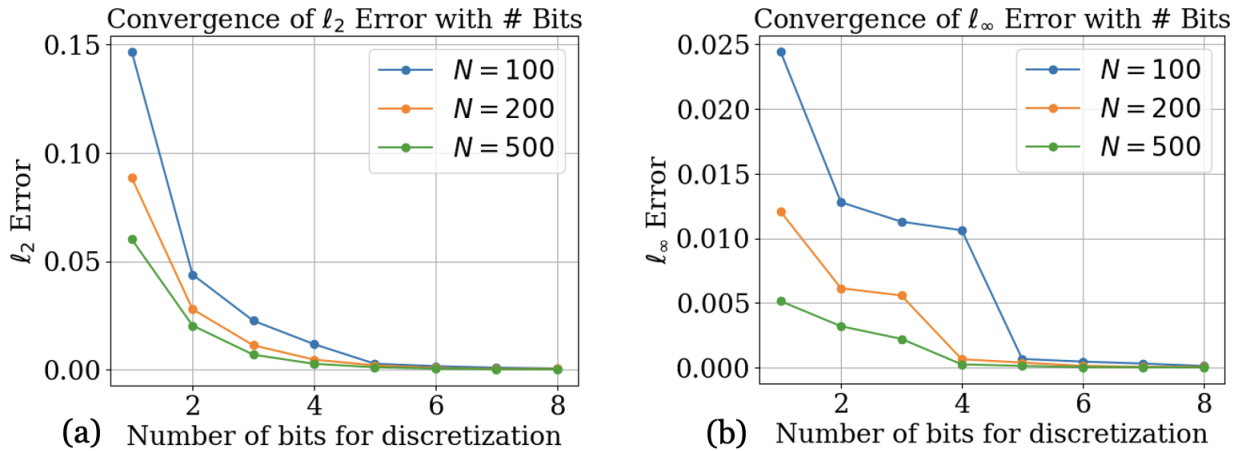}
\caption{Convergence of the discretization error with the number of bits growth. The $y$-axis shows the $\ell_2$ or $\ell_\infty$ norm of the difference between the Shapley values computed based on $b$ bits and $b+1$ bits. The lower, the better. We use Fraud dataset from OpenML \cite{dal2015calibrating}, and we use $K=5$ here.}
\label{fig:discretization_error}
\end{figure}

\begin{figure}[h]
    \centering
    \includegraphics[width=0.54\columnwidth]{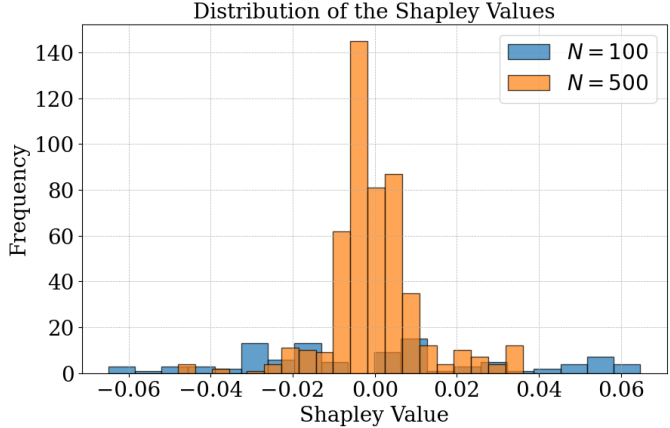}
    \caption{
    Distributions of WKNN-Shapley on different sizes of the subset of Fraud dataset from OpenML \cite{dal2015calibrating} (the number of bits for discretization $b = 5$ and $K = 5$).
    }
    \label{fig:shapley-distribution}
\end{figure}

\subsubsection{Performance on discerning data quality}

In this section, we compare the performance between continuous and discretized weighted KNN-Shapley on the tasks of discerning data quality, specifically mislabeled data detection and noisy data detection. 
The experiment settings are the same as those described in Section \ref{sec:eval-application} and Appendix \ref{appendix:eval-noisydetect}. 

\paragraph{Comparison between continuous and discretized WKNN-Shapley.}
Due to the runtime complexity of $O(N^K)$ associated with the exact algorithm for computing WKNN-Shapley with continuous weights, it becomes unfeasible to compute WKNN-Shapley with continuous weights for values of $K > 2$. Thus, in this section, our focus remains on the comparative performance for $K=2$.
Table \ref{tb:cts_vs_disc_K2-mislabel} shows the performance comparison on mislabeled data detection task, and Table \ref{tb:cts_vs_disc_K2-noisy} shows the performance comparison on noisy data detection task (with the same setting as stated in Appendix \ref{appendix:eval-settings}). 
As we can see, WKNN-Shapley with discretized weights (third column) maintains a performance closely aligned with its counterpart that uses continuous weights (first column). This observation further validates that weight discretization only has a small influence on the efficacy of WKNN-Shapley in differentiating between high- and low-quality data points.

\paragraph{Comparison of discretized weighted KNN-Shapley with different discretization bits.} 
Figure \ref{fig:varyB-mislabel} and \ref{fig:varyB-noisy} show the performance comparison of WKNN-Shapley with different number of bits for discretization $b$, on the task of mislabeled data detection and noisy data detection, respectively. 
As we can see, the performance of both the exact and deterministic approximation of WKNN-Shapley is relatively stable, regardless of the number of discretization bits.

\begin{table}[h]
\centering
\setlength\intextsep{0pt}
\setlength\abovecaptionskip{2pt}
\setlength\belowcaptionskip{0pt}
\resizebox{\columnwidth}{!}{\input{cts_vs_disc_K2-mislabel}}
\caption{AUROC scores of different variants of weighted KNN-Shapley for mislabeled data detection on benchmark datasets at $K=2$. 
The higher the AUROC score is, the better the method is. 
The Monte Carlo approximation (second column) is a stochastic algorithm. However, since it is already computationally expensive for just a single execution, we only run it once. 
}
\label{tb:cts_vs_disc_K2-mislabel}
\end{table}

\begin{table}[h]
\centering
\setlength\intextsep{0pt}
\setlength\abovecaptionskip{2pt}
\setlength\belowcaptionskip{0pt}
\resizebox{\columnwidth}{!}{\input{cts_vs_disc_K2-noisy}}
\caption{AUROC scores of different variants of weighted KNN-Shapley for noisy data detection on benchmark datasets at $K=2$. 
The higher the AUROC score is, the better the method is. 
The Monte Carlo approximation (second column) is a stochastic algorithm. However, since it is already computationally expensive for just a single execution, we only run it once. 
}
\label{tb:cts_vs_disc_K2-noisy}
\end{table}

\begin{figure}[h]
    \centering
    \setlength\intextsep{0pt}
    \setlength\abovecaptionskip{0pt}
    \setlength\belowcaptionskip{0pt}
    \centering
    \includegraphics[width=0.8\columnwidth]{image/varyB_Mislabel.pdf}
    \caption{
    AUROC scores of different variants of KNN-Shapley for noisy data detection with different discretization bits 
    $b$. The higher the curve is, the better the method is. 
    }
    \label{fig:varyB-mislabel}
\end{figure}

\begin{figure}[h]
    \centering
    \setlength\intextsep{0pt}
    \setlength\abovecaptionskip{0pt}
    \setlength\belowcaptionskip{0pt}
    \centering
    \includegraphics[width=0.8\columnwidth]{image/varyB_Noisy.pdf}
    \caption{
    AUROC scores of different variants of KNN-Shapley for noisy data detection with different discretization bits 
    $b$. 
    The higher the curve is, the better the method is. 
    }
    \label{fig:varyB-noisy}
\end{figure}

\clearpage

\subsection{Settings \& 
Additional Experiments for Runtime Comparison}
\label{appendix:runtime}

\textbf{Detailed Settings.} 
For the runtime comparison experiment in Section \ref{sec:eval-runtime}, we follow similar experiment settings from prior study \cite{kwon2023data} and use a synthetic binary classification dataset. To generate the synthetic dataset, we sample data points from a $2$-dimensional standard Gaussian distribution, and the labels are assigned based on the sign of the sum of the two features. We note that the dataset dimension has minimal impact on the runtime of WKNN-Shapley compared with dataset size $N$, since the dataset dimension only affects the runtime of computing the distance between data points. All experiments were conducted on a 32-Core 2.6 GHz Intel Skylake CPU Processor.

We present additional experimental results comparing runtimes, further expanding on Section \ref{sec:eval-runtime}. We vary both $K$, the KNN hyperparameter, and $b$, the bit count for discretization. 

Figure \ref{fig:runtime-varyB} shows the runtime comparison for our exact method and deterministic approximation for WKNN-Shapley, considering different number of bits $b$ for discretization. As we can see, although the runtime increases with more bits for discretization, our algorithms, even at $b=7$, demonstrate a $>10^4$ times of improvement over both the exact computation and approximation algorithm introduced in \cite{jia2019efficient}. 

Figure \ref{fig:runtime-varyK-exact} shows the runtime comparison between our exact WKNN-Shapley computation algorithm and the $O(N^K)$ algorithm from \cite{jia2019efficient}, considering different choices of $K$. Since $K=10$ for the baseline algorithm is computationally infeasible even for very small $N$ (e.g., 20), we do not show the curve here. As we can see, our algorithm's curves for $K=3$, $K=5$, and $K=10$ exhibit a relatively modest ascent in runtime, staying well below $10^6$ seconds even at 100,000 training points. In contrast, the algorithm from \cite{jia2019efficient} witnesses a steeper rise. This distinction is expected given that our algorithm's runtime scales at $O(K^2)$, whereas the one from \cite{jia2019efficient} features an exponential time complexity with respect to $K$. 
Figure \ref{fig:runtime-varyK-approx} shows the runtime comparison between our deterministic approximation algorithm and the Monte Carlo-based approximation algorithm (from \cite{jia2019efficient}) for WKNN-Shapley. As we can see, our approximation algorithm is around $> 10^4$ times faster than the Monte Carlo algorithm for achieving the same error bound. 

\begin{remark}
We note that, as a training-free algorithm, KNN-Shapley exhibits a significant advantage in its computational efficiency compared with approaches such as Data Shapley/Banzhaf which requires many model retraining. For instance, as reported in Data Banzhaf's official Github repo\footnote{https://github.com/Jiachen-T-Wang/data-banzhaf}, it takes around 5 CPU hours to train 10,000 very small MLP models on different subsets of a tiny, size-200 dataset! On the contrary, it only takes a few seconds for KNN-Shapley under the same setting.
\end{remark}


\begin{figure}[h]
    \centering
    \centering
    \includegraphics[width=0.85\columnwidth]{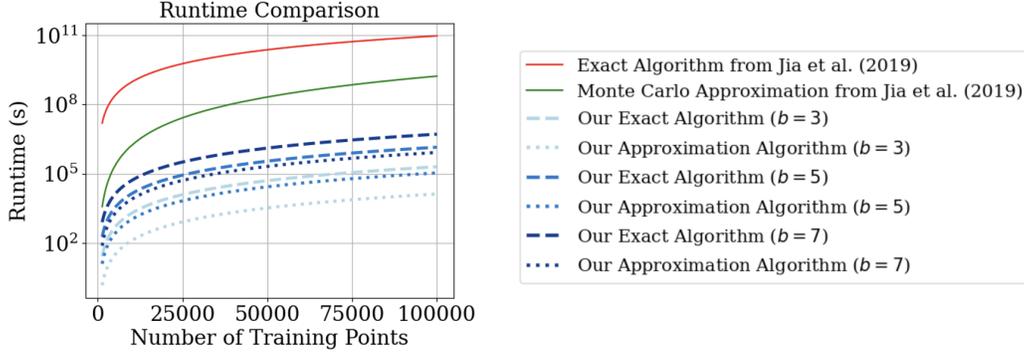}
    \caption{
    Runtime comparison between our exact and approximation algorithms for WKNN-Shapley in Section \ref{sec:shapley-for-binary}, and those from \cite{jia2019efficient}, across varying training data sizes $N$. We set $K = 5$ here for all methods. 
    For our algorithms from Section \ref{sec:shapley-for-binary}, we vary the number of bits for discretization. 
    All other settings are the same as Figure \ref{fig:runtime-nb3} in the maintext. 
    }
    \label{fig:runtime-varyB}
\end{figure}

\begin{figure}[h]
    \centering
    \centering
    \includegraphics[width=0.85\columnwidth]{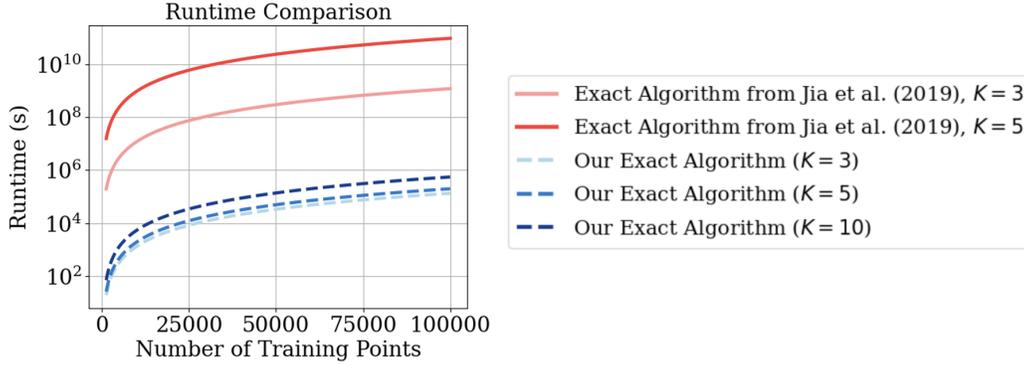}
    \caption{
    Runtime comparison between our exact WKNN-Shapley computation algorithm in Section \ref{sec:exact-shapley}, and the $O(N^K)$ from \cite{jia2019efficient}, across varying training data sizes $N$. 
    We set $b=3$ here for weights discretization in our algorithm. 
    We vary and compare the runtime for different choices of $K$. 
    All other settings are the same as Figure \ref{fig:runtime-nb3} in the maintext. 
    }
    \label{fig:runtime-varyK-exact}
\end{figure}

\begin{figure}[h]
    \centering
    \centering
    \includegraphics[width=0.85\columnwidth]{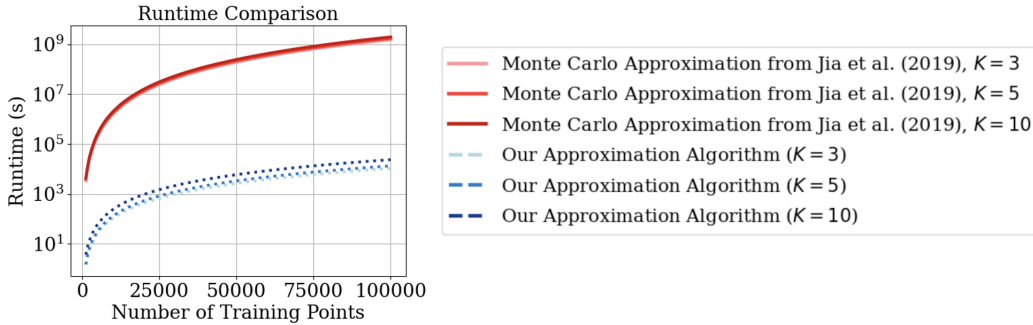}
    \caption{
    Runtime comparison between our deterministic WKNN-Shapley approximation algorithm in Section \ref{sec:deterministic-approx}, and the Monte Carlo approximation from \cite{jia2019efficient}, across varying training data sizes $N$. 
    We set $b=3$ here for weights discretization in our algorithm. 
    We vary and compare the runtime for different choices of $K$. 
    All other settings are the same as Figure \ref{fig:runtime-nb3} in the maintext. 
    }
    \label{fig:runtime-varyK-approx}
\end{figure}

\clearpage

\subsection{Application on Noisy Data Detection}
\label{appendix:eval-noisydetect}

\textbf{Settings.} 
In the experiment of noisy data detection, we randomly choose 10\% of the data points and add strong noise to their features. Specifically, we add zero-mean Gaussian noise to data features, where the standard deviation of the Gaussian noise added to each feature dimension is equal to the average absolute value of the feature dimension across the full dataset. 
Similar to the task of mislabeled data detection, we use AUROC as the performance metric on noisy data detection tasks. 

Table \ref{tb:noisy-detection} shows the AUROC scores across the 13 benchmark datasets we experimented on when $K=5$ and number of bits for discretization $b=3$. 
Similar to the results for mislabeled data detection, we can see that both exact and approximated WKNN-Shapley significantly outperform the unweighted KNN-Shapley (either soft-label or hard-label) across most datasets, attributable to WKNN-Shapley’s ability to more accurately differentiate between bad and good data based on the additional information of the proximity to the queried example. 
We can also see the similar encouraging result that the approximated WKNN-Shapley achieves performance comparable to, and sometimes even slightly better than, the exact WKNN-Shapley across the majority of datasets. This is likely attributable to its favored property in preserving the fairness properties of its exact counterpart. 

In Figure \ref{fig:ablation-varyK-noisy}, we show similar result on the task of noisy data detection that, compared to unweighted KNN-Shapley, WKNN-Shapley maintains notably stable performance across various choices of $K$, particularly for larger values of $K$. 
This is attributable to the additional weighting information incorporated in WKNN-Shapley. 


\begin{table}[h]
\centering
\resizebox{\columnwidth}{!}{\input{noisy_detect_table}}
\caption{AUROC scores of different variants of KNN-Shapley for noisy data detection tasks on various datasets. The higher the AUROC score is, the better the method is.
}
\label{tb:noisy-detection}
\end{table}

\begin{figure}[h]
    \centering
    \centering
    \includegraphics[width=\columnwidth]{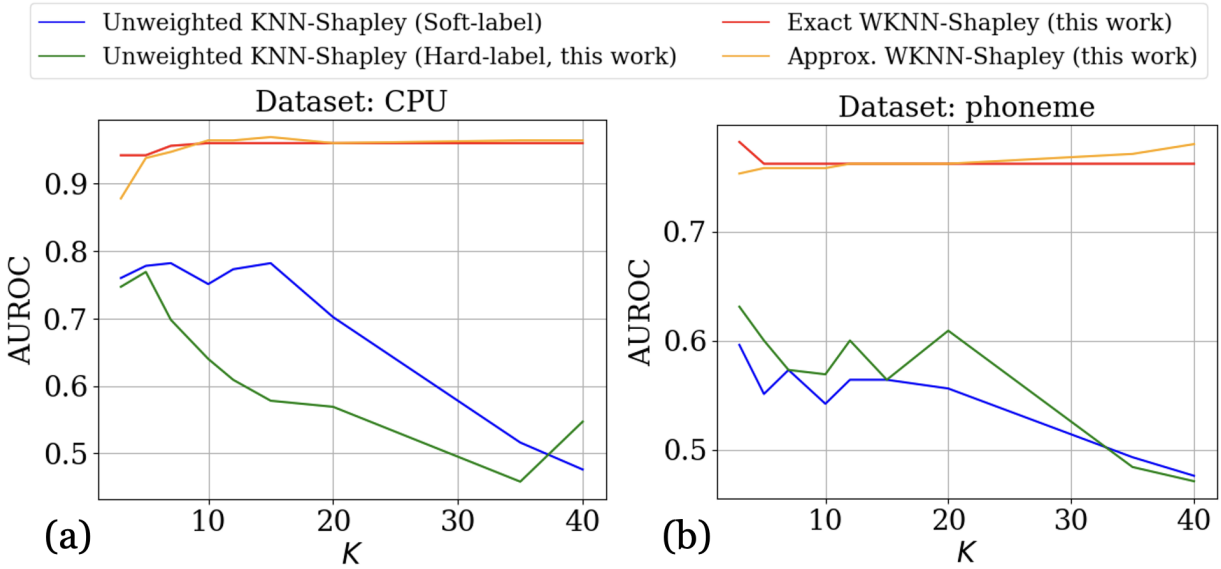}
    \caption{
    AUROC scores of different variants of KNN-Shapley for noisy data detection with different $K$s. The higher the curve is, the better the method is. 
    }
    \label{fig:ablation-varyK-noisy}
\end{figure}

\clearpage

\subsection{Ablation Study for the choice of $\mstar$ for the deterministic approximation algorithm}
\label{appendix:eval-varyEPS}

In this section, we evaluate the performance variation when we pick different $\mstar$ for our deterministic approximation algorithm in Section \ref{sec:deterministic-approx}. 
Specifically, we choose different values of $\mstar$ and plot the performance variation in mislabeled/noisy data detection task with different error bound $\eps(\mstar)$. 
Note that $\eps(\mstar) = 0$ corresponds to the exact WKNN-Shapley. 
We also highlight the location of $\eps(\sqrt{N})$, i.e., the error bound for the $\mstar$ we set in the experiment. 
As we can see from Figure \ref{fig:varyEPS}, the performance of the approximated WKNN-Shapley is highly stable across a wide range of choices of $\mstar$s. 
Hence, we set $\mstar = \sqrt{N}$ in all of the experiments instead of following the adaptive procedure of selecting $\mstar$ mentioned in Appendix \ref{appendix:how-to-select-mstar}.

\begin{figure}[h]
    \centering
    \includegraphics[width=\columnwidth]{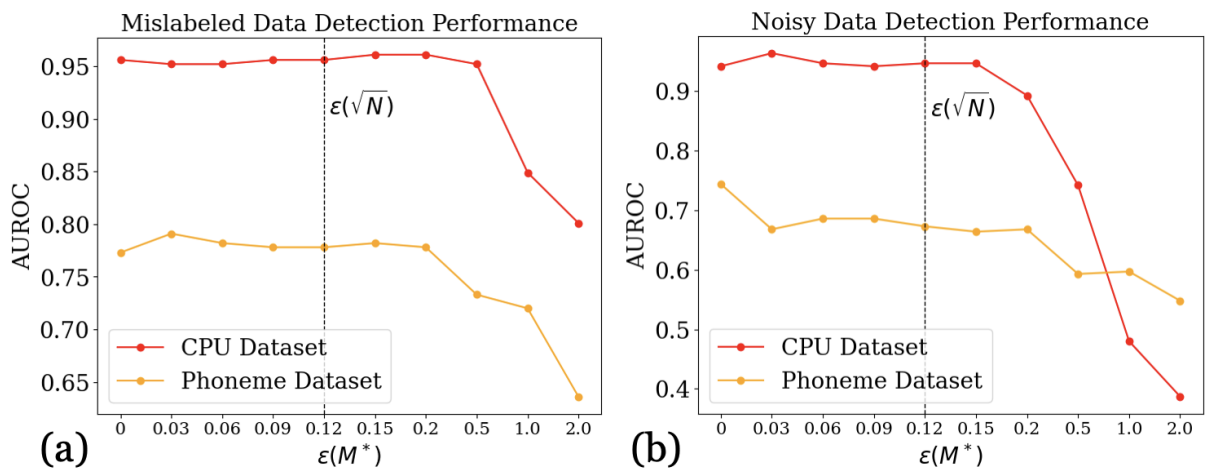}
    \caption{Performance variation of approximated WKNN-Shapley on (a) mislabeled data detection and (b) noisy data detection across different choice of $\mstar$. For a more direct comparison, we plot x-axis as the theoretical error bound $\eps(\mstar)$ derived in Theorem \ref{thm:error-bound}.}
    \label{fig:varyEPS}
\end{figure}

\clearpage

\subsection{Qualitative Comparison between Weighted and Unweighted KNN-Shapley}
\label{appendix:eval-comparison}

In our experiments, we show that weighted KNN-Shapley significantly outperforms unweighted KNN-Shapley in discerning data quality. This can likely be attributed to WKNN-Shapley’s adept ability to more accurately differentiate between bad and good data based on the proximity to the queried example. In this section, we present a more detailed qualitative analysis highlighting why WKNN-Shapley outperforms unweighted KNN-Shapley in discerning data quality. 

Figure \ref{fig:qualitative} shows the value score distribution of unweighted and weighted KNN-Shapley of 50 data points, where 5 of them are being mislabeled. 
The KNN-Shapley scores are computed with respect to a single validation point. 
As we can see, compared with the unweighted KNN-Shapley, WKNN-Shapley exhibits a much higher variation in value scores of different data points. 
More importantly, WKNN-Shapley can better differentiate the quality between the data points that are near the validation point. 
As we can see from the figure, the two mislabeled points that are the closest to the validation point (index \textbf{49} and \textbf{50}) receive much lower WKNN-Shapley scores compared with those benign data points that have a different label as the validation point (index \textbf{41-48}). 
On the other hand, unweighted KNN-Shapley assigns almost the same negative values for all data points that are close to the validation point but has a different label (index \textbf{41-50}), regardless of whether they are benign or mislabeled. 

\begin{figure}[h]
    \centering
    \includegraphics[width=\columnwidth]{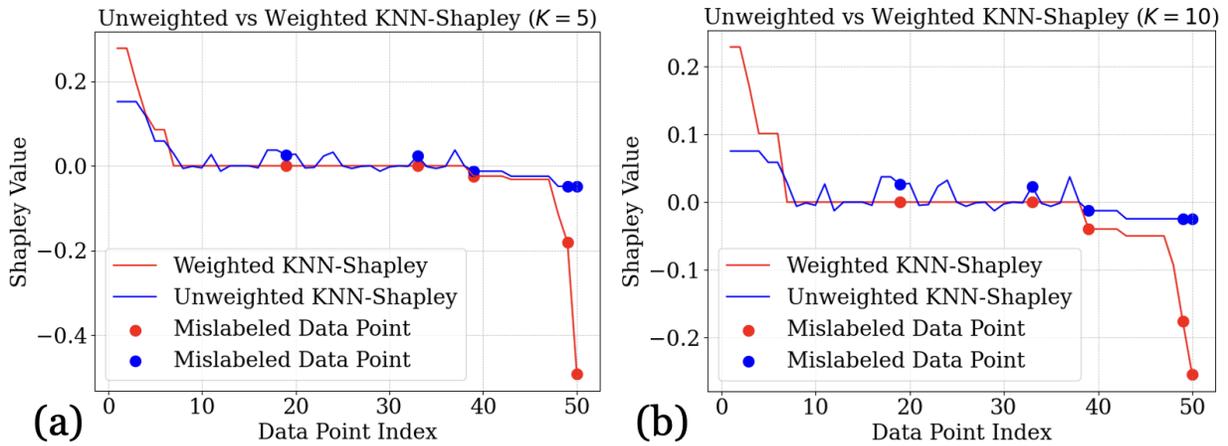}
    \caption{Distribution of unweighted and weighted KNN-Shapley scores at (a) $K=5$ and (b) $K=10$ for 50 data points from CPU dataset, where 5 of them are mislabeled data points.}
    \label{fig:qualitative}
\end{figure}

\clearpage

\subsection{Application on Data Selection for Neural Network Training and KNN-LM}
\label{appendix:eval-dataselection}

In this section, we showcase an additional application of weighted KNN-Shapley in selecting high-quality data points for important machine learning tasks. 
Specifically, we evaluate the utility of selected data points for (1) regular neural network training, and (2) $K$ nearest neighbor language models (KNN-LMs) \cite{khandelwal2019generalization}, a famous type of retrieval-augmented language model nowadays.



\begin{remark}[\textbf{Data selection is not the major application of data valuation techniques}]
We note that for the task of data selection, the performance is highly dependent on the diversity of selected data points. However, a reasonable data value notion typically needs to satisfy Symmetry axiom, which is being interpreted as fairness. Hence, if one selects data points with high values, it is likely that the selected data points lack diversity, as similar data points are required to receive similar value scores. In fact, preliminary theoretical and empirical results \cite{wang2021towards} show that the Shapley value as well as other cooperative-game-theory-based data valuation can result in arbitrarily low model performance in the worst case. 
Hence, we do not consider data selection as the major application of data valuation techniques. 
Having said that, in this section, we evaluate the performance of WKNN-Shapley on data selection task and demonstrate its superior performance compared with unweighted KNN-Shapley. 
\end{remark}

Figure \ref{fig:selection} (a) shows WKNN's performance on CIFAR10 \cite{krizhevsky2009learning} when trained on data points that receive the highest data value scores (computed based on the associated data valuation techniques). 
We use ResNet18 with batch size $128$, (initial) learning rate $10^{-3}$ and Adam optimizer for training. 
Since CIFAR10 is a relatively well-curated dataset, we manually introduce quality variation by randomly flipping the labels of 10\% of the training images. 
As we can see, both the exact and approximated WKNN-Shapley offer comparable results, and both of them outperform the original unweighted KNN-Shapley. 

Figure \ref{fig:selection} (b) shows KNN-LM's performance on the WNLI dataset \cite{wang2018glue}, where the \emph{data store} incorporates only those data points that receive the highest value scores. 
Since WNLI dataset does not have a publically available test set, we split its original validation set and pick 25 validation points as the validation set, and we use the rest of the validation points as the test set. 
KNN-LM is a popular retrieval-augmented language model where the output of the original LM is being interpolated with the output of the KNN classifiers, i.e., $p_{KNN-LM}(y) := \lambda p_{KNN}(y) + (1-\lambda) p_{LM}(y)$. Here, we set $\lambda = 0.5$. We use BERT \cite{kenton2019bert} pretrained on WNLI dataset as the language model here.\footnote{Publically available from \url{https://huggingface.co/JeremiahZ/bert-base-uncased-wnli} with a baseline accuracy around 53\% on splited test set.} 
As we can see from the figure, both the exact and approximated WKNN-Shapley stand out and outperform the original unweighted KNN-Shapley by a large margin. We note that when leveraging $>55\%$ of the entire data store, KNN-LM performs even worse than the original, unaugmented LM due to the relatively low quality of the benchmark dataset. This underscores the important role of selecting high-quality data points, where WKNN-Shapley proves to be an effective tool. 

\begin{figure}[h]
    \centering
    \includegraphics[width=0.8\columnwidth]{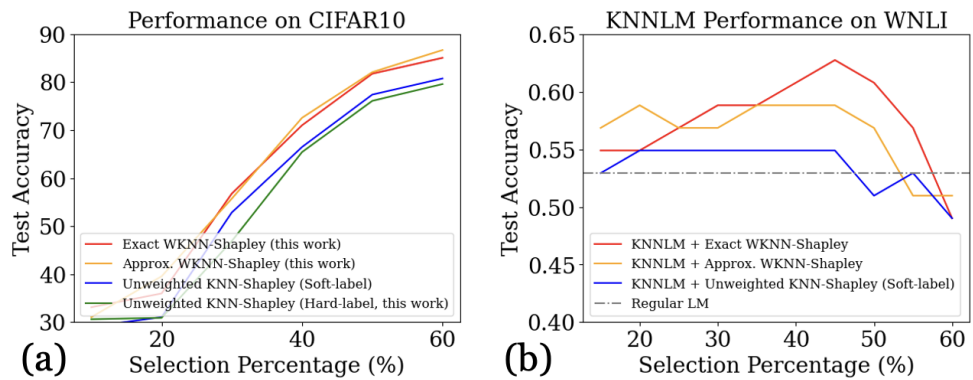}
    \caption{
    (a) The performance of ResNet18 on the CIFAR10 subset selected by different data valuation techniques. 
    (b) The performance of KNN-LM on the WNLI dataset \cite{wang2018glue}'s subset selected by different data valuation techniques. 
    }
    \label{fig:selection}
\end{figure}

%% file: cts_vs_disc_K2-mislabel.tex
\begin{tabular}{@{}ccccc@{}}
\toprule
                    & \textbf{\begin{tabular}[c]{@{}c@{}}Weighted\\ KNN-Shapley \\ (Continuous Weights)\end{tabular}} & \textbf{\begin{tabular}[c]{@{}c@{}}Weighted\\ KNN-Shapley\\ (Continuous Weights, \\ Monte Carlo Approximation)\end{tabular}} & \textbf{\begin{tabular}[c]{@{}c@{}}Weighted\\ KNN-Shapley\\ (Discretized Weights, \\ this work)\end{tabular}} & \textbf{\begin{tabular}[c]{@{}c@{}}Weighted\\ KNN-Shapley\\ (Discretized Weights, \\ deterministic approximation, \\ this work)\end{tabular}} \\ \midrule
\textbf{2DPlanes}   & 0.853                                                                                           & 0.851                                                                                                                        & 0.854                                                                                                         & 0.851                                                                                                                                         \\
\textbf{CPU}        & 0.809                                                                                           & 0.791                                                                                                                        & 0.809                                                                                                         & 0.796                                                                                                                                         \\
\textbf{Phoneme}    & 0.729                                                                                           & 0.609                                                                                                                        & 0.747                                                                                                         & 0.742                                                                                                                                         \\
\textbf{Fraud}      & 0.511                                                                                           & 0.501                                                                                                                        & 0.507                                                                                                         & 0.513                                                                                                                                         \\
\textbf{Creditcard} & 0.723                                                                                           & 0.716                                                                                                                        & 0.724                                                                                                         & 0.698                                                                                                                                         \\
\textbf{Vehicle}    & 0.733                                                                                           & 0.636                                                                                                                        & 0.72                                                                                                          & 0.769                                                                                                                                         \\
\textbf{Click}      & 0.707                                                                                           & 0.609                                                                                                                        & 0.707                                                                                                         & 0.719                                                                                                                                         \\
\textbf{Wind}       & 0.81                                                                                            & 0.769                                                                                                                        & 0.813                                                                                                         & 0.809                                                                                                                                         \\
\textbf{Pol}        & 0.973                                                                                           & 0.978                                                                                                                        & 0.996                                                                                                         & 1                                                                                                                                             \\
\textbf{MNIST}      & 0.732                                                                                           & 0.742                                                                                                                        & 0.733                                                                                                         & 0.72                                                                                                                                          \\
\textbf{CIFAR10}    & 0.742                                                                                           & 0.64                                                                                                                         & 0.729                                                                                                         & 0.722                                                                                                                                         \\ 
\textbf{AGNews}     & 0.942                                                                                        & 0.932                                                                                     & 0.944                                                                                  & 0.914                                                                                         \\
\textbf{DBPedia}    & 0.969                                                                                        & 0.969                                                                                     & 0.988                                                                                  & 0.988 \\ \bottomrule
\end{tabular}

%% file: cts_vs_disc_K2-noisy.tex
\begin{tabular}{@{}ccccc@{}}
\toprule
                    & \textbf{\begin{tabular}[c]{@{}c@{}}Weighted\\ KNN-Shapley \\ (Continuous Weights)\end{tabular}} & \textbf{\begin{tabular}[c]{@{}c@{}}Weighted\\ KNN-Shapley\\ (Continuous Weights, \\ Monte Carlo Approximation)\end{tabular}} & \textbf{\begin{tabular}[c]{@{}c@{}}Weighted\\ KNN-Shapley\\ (Discretized Weights, \\ this work)\end{tabular}} & \textbf{\begin{tabular}[c]{@{}c@{}}Weighted\\ KNN-Shapley\\ (Discretized Weights, \\ deterministic approximation, \\ this work)\end{tabular}} \\ \midrule
\textbf{2DPlanes}   & 0.62                                                                                            & 0.524                                                                                                                        & 0.609                                                                                                         & 0.604                                                                                                                                         \\
\textbf{CPU}        & 0.76                                                                                            & 0.698                                                                                                                        & 0.8                                                                                                           & 0.676                                                                                                                                         \\
\textbf{Phoneme}    & 0.591                                                                                           & 0.542                                                                                                                        & 0.578                                                                                                         & 0.6                                                                                                                                           \\
\textbf{Fraud}      & 0.831                                                                                           & 0.722                                                                                                                        & 0.831                                                                                                         & 0.622                                                                                                                                         \\
\textbf{Creditcard} & 0.493                                                                                           & 0.507                                                                                                                        & 0.533                                                                                                         & 0.489                                                                                                                                         \\
\textbf{Vehicle}    & 0.569                                                                                           & 0.529                                                                                                                        & 0.533                                                                                                         & 0.44                                                                                                                                          \\
\textbf{Click}      & 0.431                                                                                           & 0.458                                                                                                                        & 0.413                                                                                                         & 0.333                                                                                                                                         \\
\textbf{Wind}       & 0.773                                                                                           & 0.667                                                                                                                        & 0.773                                                                                                         & 0.649                                                                                                                                         \\
\textbf{Pol}        & 0.502                                                                                           & 0.52                                                                                                                         & 0.547                                                                                                         & 0.471                                                                                                                                         \\
\textbf{MNIST}      & 0.68                                                                                            & 0.671                                                                                                                        & 0.676                                                                                                         & 0.678                                                                                                                                         \\
\textbf{CIFAR10}    & 0.502                                                                                           & 0.529                                                                                                                        & 0.509                                                                                                         & 0.502                                                                                                                                         \\ 
\textbf{AGNews}     & 0.508                                                                                        & 0.508                                                                                     & 0.472                                                                                  & 0.444                                                                                         \\
\textbf{DBPedia}    & 0.447                                                                                        & 0.451                                                                                     & 0.443                                                                                  & 0.514 \\ \bottomrule
\end{tabular}

%% file: noisy_detect_table.tex
\begin{tabular}{@{}ccccc@{}}
\toprule
& \textbf{\begin{tabular}[c]{@{}c@{}}Soft-label\\ KNN-Shapley \\ (Jia et al. (2019))\end{tabular}} & \textbf{\begin{tabular}[c]{@{}c@{}}Hard-label \\ KNN-Shapley \\ (this work)\end{tabular}} & \textbf{\begin{tabular}[c]{@{}c@{}}Exact\\ WKNN-Shapley\\ (this work)\end{tabular}} & \textbf{\begin{tabular}[c]{@{}c@{}}Approximated\\ WKNN-Shapley\\ (this work)\end{tabular}} \\ \midrule
\textbf{2DPlanes}   & 0.556                                                                                      & 0.498                                                                                      & \textbf{0.733}                                                                              & 0.68                                                                                      \\
\textbf{CPU}        & 0.778                                                                                      & 0.769                                                                                      & 0.942                                                                              & \textbf{0.947}                                                                                     \\
\textbf{Phoneme}    & 0.551                                                                                      & 0.6                                                                                        & \textbf{0.744}                                                                              & 0.673                                                                                     \\
\textbf{Fraud}      & 0.862                                                                                      & 0.858                                                                                      & 0.911                                                                              & \textbf{0.916}                                                                                     \\
\textbf{Creditcard} & 0.453                                                                                      & 0.422                                                                                      & \textbf{0.653}                                                                              & 0.636                                                                                     \\
\textbf{Vehicle}    & 0.511                                                                                      & 0.52                                                                                       & 0.916                                                                              & \textbf{0.933}                                                                                     \\
\textbf{Click}      & 0.44                                                                                       & 0.444                                                                                      & \textbf{0.711}                                                                              & 0.662                                                                                     \\
\textbf{Wind}       & 0.782                                                                                      & 0.804                                                                                      & 0.849                                                                              & \textbf{0.853}                                                                                     \\
\textbf{Pol}        & 0.493                                                                                      & 0.502                                                                                      & \textbf{0.836}                                                                              & 0.804                                                                                     \\
\textbf{MNIST}      & 0.782                                                                                      & 0.538                                                                                      & \textbf{0.911}                                                                              & \textbf{0.911}                                                                                     \\
\textbf{CIFAR10}    & 0.533                                                                                      & 0.418                                                                                      & 0.8                                                                                & \textbf{0.822}    \\      
\textbf{AGNews}     & 0.481                                                                                           & 0.531                                                                                                                        & \textbf{0.559}                                                                                                         & 0.543                                                                                                                                         \\
\textbf{DBPedia}    & 0.482                                                                                           & 0.498                                                                                                                        & \textbf{0.58}                                                                                                          & 0.576                                                                                                                                         \\ \bottomrule

\end{tabular}

%% file: arxiv.bbl
\begin{thebibliography}{}

\bibitem[Anonymous, 2021]{wang2021towards}
Anonymous (2021).
\newblock Towards general robustness to bad training data.

\bibitem[Auer et~al., 2007]{auer2007dbpedia}
Auer, S., Bizer, C., Kobilarov, G., Lehmann, J., Cyganiak, R., and Ives, Z.
  (2007).
\newblock Dbpedia: A nucleus for a web of open data.
\newblock In {\em The Semantic Web: 6th International Semantic Web Conference,
  2nd Asian Semantic Web Conference, ISWC 2007+ ASWC 2007, Busan, Korea,
  November 11-15, 2007. Proceedings}, pages 722--735. Springer.

\bibitem[Belaid et~al., 2023]{belaid2023optimizing}
Belaid, M.~K., Mekki, D.~E., Rabus, M., and H{\"u}llermeier, E. (2023).
\newblock Optimizing data shapley interaction calculation from o (2\^{} n) to o
  (tn\^{} 2) for knn models.
\newblock {\em arXiv preprint arXiv:2304.01224}.

\bibitem[Bian et~al., 2021]{bian2021energy}
Bian, Y., Rong, Y., Xu, T., Wu, J., Krause, A., and Huang, J. (2021).
\newblock Energy-based learning for cooperative games, with applications to
  valuation problems in machine learning.
\newblock {\em arXiv preprint arXiv:2106.02938}.

\bibitem[Courtnage and Smirnov, 2021]{courtnage2021shapley}
Courtnage, C. and Smirnov, E. (2021).
\newblock Shapley-value data valuation for semi-supervised learning.
\newblock In {\em Discovery Science: 24th International Conference, DS 2021,
  Halifax, NS, Canada, October 11--13, 2021, Proceedings 24}, pages 94--108.
  Springer.

\bibitem[Dal~Pozzolo et~al., 2015]{dal2015calibrating}
Dal~Pozzolo, A., Caelen, O., Johnson, R.~A., and Bontempi, G. (2015).
\newblock Calibrating probability with undersampling for unbalanced
  classification.
\newblock In {\em 2015 IEEE Symposium Series on Computational Intelligence},
  pages 159--166. IEEE.

\bibitem[Dall’Aglio et~al., 2019]{dall2019sometimes}
Dall’Aglio, M., Fragnelli, V., and Moretti, S. (2019).
\newblock Sometimes the computation of the shapley value is simple.
\newblock {\em Handbook of the Shapley value}, 441.

\bibitem[Ghorbani et~al., 2020]{ghorbani2020distributional}
Ghorbani, A., Kim, M., and Zou, J. (2020).
\newblock A distributional framework for data valuation.
\newblock In {\em International Conference on Machine Learning}, pages
  3535--3544. PMLR.

\bibitem[Ghorbani and Zou, 2019]{ghorbani2019data}
Ghorbani, A. and Zou, J. (2019).
\newblock Data shapley: Equitable valuation of data for machine learning.
\newblock In {\em International Conference on Machine Learning}, pages
  2242--2251. PMLR.

\bibitem[Ghorbani et~al., 2022]{ghorbani2022data}
Ghorbani, A., Zou, J., and Esteva, A. (2022).
\newblock Data shapley valuation for efficient batch active learning.
\newblock In {\em 2022 56th Asilomar Conference on Signals, Systems, and
  Computers}, pages 1456--1462. IEEE.

\bibitem[He et~al., 2016]{he2016deep}
He, K., Zhang, X., Ren, S., and Sun, J. (2016).
\newblock Deep residual learning for image recognition.
\newblock In {\em Proceedings of the IEEE conference on computer vision and
  pattern recognition}, pages 770--778.

\bibitem[Jia et~al., 2019a]{jia2019efficient}
Jia, R., Dao, D., Wang, B., Hubis, F.~A., Gurel, N.~M., Li, B., Zhang, C.,
  Spanos, C.~J., and Song, D. (2019a).
\newblock Efficient task-specific data valuation for nearest neighbor
  algorithms.
\newblock {\em Proceedings of the VLDB Endowment}.

\bibitem[Jia et~al., 2019b]{jia2019towards}
Jia, R., Dao, D., Wang, B., Hubis, F.~A., Hynes, N., G{\"u}rel, N.~M., Li, B.,
  Zhang, C., Song, D., and Spanos, C.~J. (2019b).
\newblock Towards efficient data valuation based on the shapley value.
\newblock In {\em The 22nd International Conference on Artificial Intelligence
  and Statistics}, pages 1167--1176. PMLR.

\bibitem[Just et~al., 2022]{just2022lava}
Just, H.~A., Kang, F., Wang, T., Zeng, Y., Ko, M., Jin, M., and Jia, R. (2022).
\newblock Lava: Data valuation without pre-specified learning algorithms.
\newblock In {\em The Eleventh International Conference on Learning
  Representations}.

\bibitem[Karla{\v{s}} et~al., 2022]{karlavs2022data}
Karla{\v{s}}, B., Dao, D., Interlandi, M., Li, B., Schelter, S., Wu, W., and
  Zhang, C. (2022).
\newblock Data debugging with shapley importance over end-to-end machine
  learning pipelines.
\newblock {\em arXiv preprint arXiv:2204.11131}.

\bibitem[Kenton and Toutanova, 2019]{kenton2019bert}
Kenton, J. D. M.-W.~C. and Toutanova, L.~K. (2019).
\newblock Bert: Pre-training of deep bidirectional transformers for language
  understanding.
\newblock In {\em Proceedings of NAACL-HLT}, pages 4171--4186.

\bibitem[Khandelwal et~al., 2019]{khandelwal2019generalization}
Khandelwal, U., Levy, O., Jurafsky, D., Zettlemoyer, L., and Lewis, M. (2019).
\newblock Generalization through memorization: Nearest neighbor language
  models.
\newblock In {\em International Conference on Learning Representations}.

\bibitem[Krizhevsky et~al., 2009]{krizhevsky2009learning}
Krizhevsky, A., Hinton, G., et~al. (2009).
\newblock Learning multiple layers of features from tiny images.

\bibitem[Kwon and Zou, 2022]{kwon2022beta}
Kwon, Y. and Zou, J. (2022).
\newblock Beta shapley: a unified and noise-reduced data valuation framework
  for machine learning.
\newblock In {\em International Conference on Artificial Intelligence and
  Statistics}, pages 8780--8802. PMLR.

\bibitem[Kwon and Zou, 2023]{kwon2023data}
Kwon, Y. and Zou, J. (2023).
\newblock Data-oob: Out-of-bag estimate as a simple and efficient data value.
\newblock {\em ICML}.

\bibitem[LeCun, 1998]{lecun1998mnist}
LeCun, Y. (1998).
\newblock The mnist database of handwritten digits.
\newblock {\em http://yann. lecun. com/exdb/mnist/}.

\bibitem[Liang et~al., 2021]{liang2021herald}
Liang, W., Liang, K.-H., and Yu, Z. (2021).
\newblock Herald: An annotation efficient method to detect user disengagement
  in social conversations.
\newblock In {\em Proceedings of the 59th Annual Meeting of the Association for
  Computational Linguistics and the 11th International Joint Conference on
  Natural Language Processing (Volume 1: Long Papers)}, pages 3652--3665.

\bibitem[Liang et~al., 2020]{liang2020beyond}
Liang, W., Zou, J., and Yu, Z. (2020).
\newblock Beyond user self-reported likert scale ratings: A comparison model
  for automatic dialog evaluation.
\newblock In {\em Proceedings of the 58th Annual Meeting of the Association for
  Computational Linguistics}, pages 1363--1374.

\bibitem[Lin et~al., 2022]{lin2022measuring}
Lin, J., Zhang, A., L{\'e}cuyer, M., Li, J., Panda, A., and Sen, S. (2022).
\newblock Measuring the effect of training data on deep learning predictions
  via randomized experiments.
\newblock In {\em International Conference on Machine Learning}, pages
  13468--13504. PMLR.

\bibitem[Liu et~al., 2023]{liu20232d}
Liu, Z., Just, H.~A., Chang, X., Chen, X., and Jia, R. (2023).
\newblock 2d-shapley: A framework for fragmented data valuation.
\newblock {\em arXiv preprint arXiv:2306.10473}.

\bibitem[Maleki, 2015]{maleki2015addressing}
Maleki, S. (2015).
\newblock {\em Addressing the computational issues of the Shapley value with
  applications in the smart grid}.
\newblock PhD thesis, University of Southampton.

\bibitem[Northcutt et~al., 2021]{northcutt2021pervasive}
Northcutt, C.~G., Athalye, A., and Mueller, J. (2021).
\newblock Pervasive label errors in test sets destabilize machine learning
  benchmarks.
\newblock In {\em Thirty-fifth Conference on Neural Information Processing
  Systems Datasets and Benchmarks Track (Round 1)}.

\bibitem[OpenAI, 2023]{openaifuture}
OpenAI, F. (2023).
\newblock Planning for agi and beyond.
\newblock \url{ https://openai.com/blog/planning-for-agi-and-beyond }.

\bibitem[Pandl et~al., 2021]{pandl2021trustworthy}
Pandl, K.~D., Feiland, F., Thiebes, S., and Sunyaev, A. (2021).
\newblock Trustworthy machine learning for health care: scalable data valuation
  with the shapley value.
\newblock In {\em Proceedings of the Conference on Health, Inference, and
  Learning}, pages 47--57.

\bibitem[Reimers and Gurevych, 2019]{reimers-2019-sentence-bert}
Reimers, N. and Gurevych, I. (2019).
\newblock Sentence-bert: Sentence embeddings using siamese bert-networks.
\newblock In {\em Proceedings of the 2019 Conference on Empirical Methods in
  Natural Language Processing}. Association for Computational Linguistics.

\bibitem[Shapley, 1953]{shapley1953value}
Shapley, L.~S. (1953).
\newblock A value for n-person games.
\newblock {\em Contributions to the Theory of Games}, 2(28):307--317.

\bibitem[Shim et~al., 2021]{shim2021online}
Shim, D., Mai, Z., Jeong, J., Sanner, S., Kim, H., and Jang, J. (2021).
\newblock Online class-incremental continual learning with adversarial shapley
  value.
\newblock In {\em Proceedings of the AAAI Conference on Artificial
  Intelligence}, volume~35, pages 9630--9638.

\bibitem[Wang et~al., 2018]{wang2018glue}
Wang, A., Singh, A., Michael, J., Hill, F., Levy, O., and Bowman, S.~R. (2018).
\newblock Glue: A multi-task benchmark and analysis platform for natural
  language understanding.
\newblock In {\em International Conference on Learning Representations}.

\bibitem[Wang and Jia, 2023a]{wang2023data}
Wang, J.~T. and Jia, R. (2023a).
\newblock Data banzhaf: A robust data valuation framework for machine learning.
\newblock In {\em International Conference on Artificial Intelligence and
  Statistics}, pages 6388--6421. PMLR.

\bibitem[Wang and Jia, 2023b]{wang2023noteknn}
Wang, J.~T. and Jia, R. (2023b).
\newblock A note on" efficient task-specific data valuation for nearest
  neighbor algorithms".
\newblock {\em arXiv preprint arXiv:2304.04258}.

\bibitem[Wang and Jia, 2023c]{wang2023notegroup}
Wang, J.~T. and Jia, R. (2023c).
\newblock A note on" towards efficient data valuation based on the shapley
  value''.
\newblock {\em arXiv preprint arXiv:2302.11431}.

\bibitem[Wang et~al., 2023]{wang2023threshold}
Wang, J.~T., Zhu, Y., Wang, Y.-X., Jia, R., and Mittal, P. (2023).
\newblock Threshold knn-shapley: A linear-time and privacy-friendly approach to
  data valuation.
\newblock {\em arXiv preprint arXiv:2308.15709}.

\bibitem[Wang et~al., 2020]{wang2020principled}
Wang, T., Rausch, J., Zhang, C., Jia, R., and Song, D. (2020).
\newblock A principled approach to data valuation for federated learning.
\newblock In {\em Federated Learning}, pages 153--167. Springer.

\bibitem[Wang et~al., 2021]{wang2021n24news}
Wang, Z., Shan, X., Zhang, X., and Yang, J. (2021).
\newblock N24news: A new dataset for multimodal news classification.
\newblock {\em arXiv preprint arXiv:2108.13327}.

\bibitem[Warner, 2019]{dashboardact}
Warner, M. (2019).
\newblock Warner \& hawley introduce bill to force social media companies to
  disclose how they are monetizing user data.
\newblock Government Document.

\bibitem[Wu et~al., 2022]{wu2022davinz}
Wu, Z., Shu, Y., and Low, B. K.~H. (2022).
\newblock Davinz: Data valuation using deep neural networks at initialization.
\newblock In {\em International Conference on Machine Learning}, pages
  24150--24176. PMLR.

\bibitem[Yan and Procaccia, 2021]{yan2020ifyoulike}
Yan, T. and Procaccia, A.~D. (2021).
\newblock If you like shapley then you’ll love the core.
\newblock In {\em Proceedings of the AAAI Conference on Artificial
  Intelligence}, volume~35, pages 5751--5759.

\bibitem[Yeh and Lien, 2009]{yeh2009comparisons}
Yeh, I.-C. and Lien, C.-h. (2009).
\newblock The comparisons of data mining techniques for the predictive accuracy
  of probability of default of credit card clients.
\newblock {\em Expert systems with applications}, 36(2):2473--2480.

\end{thebibliography}
